\documentclass{amsart}

\usepackage[all]{xy}
\usepackage{amsaddr}
\usepackage[dvipdfmx]{hyperref}

\theoremstyle{definition}
\newtheorem{theorem}{Theorem}[section]
\newtheorem{proposition}[theorem]{Proposition}
\newtheorem{lemma}[theorem]{Lemma}
\newtheorem{corollary}[theorem]{Corollary}
\newtheorem{definition}[theorem]{Definition}
\newtheorem{example}[theorem]{Example}
\newtheorem{remark}[theorem]{Remark}

\newtheorem{claim}[theorem]{Claim}

\newtheorem{step}{Step}

\def\n{\boldsymbol{n}}
\def\e{\boldsymbol{e}}
\def\Z{\mathbb{Z}}
\def\b0{\boldsymbol{0}}

\def\mR{\mathcal{R}}
\def\m{\boldsymbol{m}}
\def\Ord{\mbox{Ord}}
\def\of{\overline{f}}
\def\oF{\overline{F}}

\makeatletter
\@addtoreset{equation}{section}
\makeatother

\usepackage[top=30mm,bottom=30mm,left=25mm,right=25mm]{geometry}

\begin{document}

\title{Coprimeness-preserving discrete KdV type equation
on an arbitrary dimensional lattice}
\author{R. Kamiya$^1$, M. Kanki$^2$, T. Mase$^1$, T. Tokihiro$^1$\\
\small $^2$ Department of Mathematics, Faculty of Engineering Science,\\
\small Kansai University, 3-3-35 Yamate, Osaka 564-8680, Japan\\
\small $^1$ Graduate School of Mathematical Sciences,\\
\small the University of Tokyo, 3-8-1 Komaba, Tokyo 153-8914, Japan
}



\begin{abstract}

We introduce an equation defined on a multi-dimensional lattice, which
can be considered as an extension to the coprimeness-preserving discrete KdV like equation in our previous paper.
The equation is also interpreted as a higher-dimensional analogue of the Hietarinta-Viallet equation, which is famous for its singularity confining property while having an exponential degree growth.
As the main theorem we prove the Laurent and the irreducibility properties
of the equation in its ``tau-function'' form.
From the theorem the coprimeness of the equation follows.
In Appendix we review the coprimeness-preserving discrete KdV like equation which
is a base equation for our main system and prove the properties such as the coprimeness.

\end{abstract}


\begingroup
\def\uppercasenonmath#1{} 
\let\MakeUppercase\relax 

\maketitle

\endgroup


\section{Introduction}

One of the first integrability tests is the singularity confinement test \cite{sc}, which
was proposed as a discrete analogue of the Painlev\'{e} test for ordinary differential equations.
The singularity confinement approach observes whether the spontaneously appearing
singularities of a discrete system disappear after a finite number of iteration steps by cancelling out.
The test is quite useful since it can easily be applied to various discrete equations.

Later, however, an example of nonintegrable discrete equations with a confined singularities has been found by Hietarinta and Viallet:
\begin{equation}
	y_{m} = y_{m-1} + \frac{a}{y^2_{m-1}} - y_{m-2}, \label{hveq_original}
\end{equation}
which is now called the Hietarinta-Viallet equation \cite{hv}.
Thus the singularity confinement test is not sufficient in its original form as an
integrability criterion.
Let us remark that there are everal successful attempts to refine the test to eliminate the insufficiency above \cite{redemption,redeeming,rod1,rod2}.

On the other hand, there is another approach to the discrete integrability
by utilizing the complexity of a given system.  
One of the first attempts in this direction is due to Arnold \cite{Arnold}, in which growth rates of the geometric properties related to the topological complexity of a diffeomorphism were estimated.
Then Veselov realized that a polynomial growth of a certain quantity such as 
the number of intersection points is closely related to the integrability of a given system \cite{Veselov}, and Falqui and Viallet observed similar relations between a polynomial growth and the existence of an invariant in the case of birational transformations of projective spaces \cite{Falqui}.
Then the notion of the algebraic entropy is presented \cite{entropy} in this stream
of complexity measurements of the iterated mappings.
The algebraic entropy $\mathcal{E}$ of a given discrete system $\{y_n\}$ is defined by
the following limit
\[
	\mathcal{E}:=\lim_{n \to +\infty} \frac{1}{n} \log (\deg y_n),
\]
which is nonnegative.
The zero algebraic entropy criterion asserts that an equation is integrable if and only if its algebraic entropy is zero.
In other words, the algebraic entropy has an information on the speed of the degree growth of $y_n$ as a rational function of the initial variables, and judges the integrability by a subexponential growth.
The algebraic entropy of the Hietarinta-Viallet equation \eqref{hveq_original} is
proved to be positive \cite{hv,takenawa}.

The empirical studies show that the algebraic entropy criterion is
quite accurate, and thus it seems natural to `define' the discrete integrability
by the zero entropy.
It should be noted that, when we study a system over a multi-dimensional lattice,
the algebraic entropy is not defined in its original form
(though attempts have been made in \cite{Viallet2006}).
In this article, a system is considered to be integrable if its iterates
have a subexponential growth in terms of its initial variables in a suitable initial configuration.
Although we use a standard type boundary condition in this paper, it is worth noting 
that the degree growth heavily depends on the initial configurations \cite{HMW}.

An equation satisfies the Laurent phenomenon (or the Laurent property) if every iterate of the equation
can be written as a Laurent polynomial of the initial variables.
Using the theory of cluster algebras, Fomin and Zelevinsky have proved that many discrete equation have this property \cite{FZ,FZ2}.
Recently, a lot of results on the relation between the cluster algebras and the integrable systems have been found: e.g.,  \cite{FH2014} by Hone et.\ al.\, \cite{Okubo2015} by Okubo, and \cite{rims,investigation} by Mase.

Though the Laurent property and the discrete integrability seem to be closely related to each other, they are not equivalent at all.
Thus, the coprimeness property, the basic idea of which is
similar to that of singularity confinement, was proposed \cite{coprimekdv, coprime, coprimetoda}.
The coprimeness property is also motivated by a transformation of dependent variables
of several discrete equations into Laurent systems (e.g., the transformation between the tau-function form and the nonlinear form of the discrete KdV equation).
The coprimeness property focuses on the cancellation of factors when the iterates
of an equation is written as rational functions of the initial variables.
An equation is said to have the coprimeness property if any factor except for monomials emerges only on a finite number of
iterates.
Recently the transformation of variables related to the coprimeness property of a given equation is utilized to obtain the algebraic entropy of the equation,
and the algebraic entropies of the Hietarinta-Viallet equation and some of its extensions are calculated \cite{extendedhv}.
By the transformation above, we obtain a tau-function type representation of the given confining non-integrable nonlinear system.
It should be noted that such a representation for the Hietarinta-Viallet equation has first been obtained by Hone in \cite{Hone2007}. Hamad et. al. named the method 
of obtaining a Laurent system from the nonlinear discrete equation by such a transformation the ``Laurentification'' \cite{Hamad}.

In this article we focus on a class of discrete equations, which pass the singularity confinement test while they are nonintegrable in terms of their exponential degree growth.
A lot of discrete equations of this type other than the Hietarinta-Viallet equation
have been found out \cite{bedfordkim, redemption, redeeming}, although they are all defined
on a one-dimensional lattice.
An example of such confining nonintegrable equations over multi-dimensional lattices had not been known for a long time.
However,
the following example was found recently \cite{qintegrable}:
\begin{equation}\label{eq:nonlinear}
	x_{t, n} = - x_{t-1, n-1} + \frac{a}{x^k_{t, n-1}} + \frac{b}{x^k_{t-1, n}}.
\end{equation}
Here $k$ is an even integer greater than one, and $a, b$ are nonzero parameters.
The equation has an evolution on the first quadrant of the $(t, n)$-plane.
Let us recall that the equation (\ref{eq:nonlinear}) has the following singularity pattern:
\begin{equation}
\begin{array}{|c||c|c|c|c}
\vdots & \vdots & \vdots & \vdots & \\ \hline
\text{init} & \text{REG} & \text{REG} & \text{REG} & \cdots \\ \hline
\text{init} & \infty^{k} & 0^1 & \text{REG} & \cdots \\ \hline
\text{init} & 0^1 & \infty^{k} & \text{REG} & \cdots \\ \hline \hline
 & \text{init} & \text{init} & \text{init} & \cdots \\ \hline
\end{array} \label{table:singpattern}
\ ,
\end{equation}
where ``$\text{init}$'' denotes a generic initial value, and ``$\text{REG}$'' a finite value depending on the initial values.
Note that even though Equation \eqref{eq:nonlinear} coincides with Hirota's discrete KdV equation \cite{Hirota1977} when $k=1$, the singularity pattern of \eqref{eq:nonlinear} for $k=1$ is different from the one in \eqref{table:singpattern}. 

The transformation between the dependent variables inferred from the above singularity pattern is as follows:
\begin{equation}
	x_{t,n} = \frac{f_{t,n} f_{t-1,n-1}}{f^k_{t-1,n} f^k_{t,n-1}}. \label{eq:xtof}
\end{equation}
We obtain the following equation by applying the change of variables \eqref{eq:xtof}
to $x_{t,n}$ of Equation (\ref{eq:nonlinear}):
\begin{equation}\label{eq:laurent}
	f_{t,n} = \frac{ - f_{t-2,n-2} f^k_{t-1,n} f^k_{t,n-1}
	+ a f^{k^2-1}_{t-1,n-1} f^{k^2}_{t,n-2} f^k_{t-1,n} f^k_{t-2,n-1}
	+ b f^{k^2-1}_{t-1,n-1} f^{k^2}_{t-2,n} f^k_{t,n-1} f^k_{t-1,n-2}
}{f^k_{t-2,n-1}f^k_{t-1,n-2}}.
\end{equation}
The Laurent and the irreducibility properties of Equation \eqref{eq:laurent}
are already known (Theorem \ref{thm:laurent}),  and we can utilize this result
to prove the coprimeness property of Equation (\ref{eq:nonlinear}) (Theorem \ref{thm:nonlinearcoprimeness}).
These results were stated in \cite{qintegrable}.

The purpose of this paper is to introduce  extensions of equations \eqref{eq:nonlinear} and \eqref{eq:laurent}
to ones defined on a higher dimensional lattice,
and to prove the properties such as the coprimeness.
In section \ref{sec:main} we introduce our main equation \eqref{pedKdV_eq}
and prove its coprimeness property.
The key idea is to transform \eqref{pedKdV_eq}
into a ``tau-function'' form \eqref{poly-form} and prove the irreducibility property of \eqref{poly-form}.
The third section is devoted to concluding remarks.
Finally, the proof of the coprimeness property of Equation \eqref{eq:nonlinear} is
elaborated on in the Appendix section since it was published only as an unrefereed report in Japanese \cite{RIMSJPN}.


\section{Multi-dimensional coprimeness-preserving equation} \label{sec:main}

\subsection{Introducing Equation \eqref{pedKdV_eq}}
A coprimeness-preserving generalization to the discrete KdV equation \eqref{eq:nonlinear} is contained as a special case in
a coprimeness-preserving equation on a higher dimensional lattice.
We introduce one of the higher dimensional lattice systems \eqref{pedKdV_eq} and prove its coprimeness property.
Moreover the Laurent property of an arbitrary reduction of the generalized ``tau-function'' form
of \eqref{pedKdV_eq} (which will be introduced later as \eqref{poly-form}) is proved using a discussion on the boundary conditions. 
Let us remark here that the irreducibility is not preserved under a reduction in general and therefore needs to be proved independently.

The main equation we introduce is the following recurrence relation:
\begin{equation}
x_{t+1,\n}+x_{t-1,\n}=\sum_{i=1}^d \left(\frac{a_i}{x_{t,\n+\e_i}^{k_i}}+\frac{b_i}{x_{t,\n-\e_i}^{l_i}}\right)
\qquad (k_i, l_i \in 2\Z_+).
\label{pedKdV_eq}
\end{equation}
Here \eqref{pedKdV_eq} is defined on the $(d+1)$-dimensional integer lattice:
$(t,\n)=(t,n_1,\cdots,n_d) \in \mathbb{Z}\times \mathbb{Z}^d=\mathbb{Z}^{d+1}$.
Each $\e_i\in\mathbb{Z}^d$ ($i=1,2,...,d$) is a unit vector: e.g., $\e_1=(1,0,\cdots,0)$, $\e_2=(0,1,0,\cdots,0)$, and $a_i, b_i$ are parameters.
We consider $t$ as the ``time'' variable and $\n=\sum_{i=1}^d n_i\e_i$ as the ``space'' variable, and the time evolution of the system \eqref{pedKdV_eq} is uniquely defined as follows:
let $I:=\{x_{-2,\n},x_{-1,\n}\}$ $(\n\in\mathbb{Z}^d)$ be the set of initial variables, then each $x_{t,\m}$ for  $t\ge 0$ and $\m\in\mathbb{Z}^d$ is uniquely calculated from the elements in $I$.
Note that in this setting  \eqref{pedKdV_eq} is invertible with respect to the time evolution: i.e., \eqref{pedKdV_eq} can be rationally solved in the opposite direction.
Also note that the evolution is only studied in the half of the whole integer lattice: i.e.,
we only study the points where the parity of $\left(t+\sum_{i=1}^d n_i\right)$ coincides with each other.
Equation \eqref{pedKdV_eq} is also considered as a generalization of the Hietarinta-Viallet equation \eqref{hveq_original} to a multi-dimensional lattice \cite{qintegrable}.

Let us investigate the singularity pattern of \eqref{pedKdV_eq}.
\begin{proposition} \label{prop:singpattern}
Let us assume the condition
\begin{equation}\label{condition}
\min_{1\le i\le d}[l_ik_i-1]>\max_{1\le i\le d}[l_i,\, k_i].
\end{equation}
Then the equation \eqref{pedKdV_eq} passes the singularity confinement test with the following pattern:
\begin{align*}
x_{0,\boldsymbol{0}}&=0^1,\\
x_{1,\e_i}&=\infty^{l_i},\ x_{1,-\e_i}=\infty^{k_i},\\
x_{2,\pm \e_i\pm \e_j}&=\mbox{REG} \ (i\neq j),\ x_{2, \pm 2\e_i}=\mbox{REG},\\
x_{2,\boldsymbol{0}}&=0^1,\\
x_{3,\pm \e_i}&=\mbox{REG},\\
x_{4,\b0}&=\mbox{REG},
\end{align*}
where REG denotes some regular value: i.e., a finite value depending on initial variables other than $x_{0,\boldsymbol{0}}$.
\end{proposition}
\begin{proof}
Let us take $x_{0,\boldsymbol{0}}=\varepsilon$ and calculate the time evolution.
Direct calculations show that
\begin{align*}
x_{1,\e_i}&=-x_{-1,\e_i}+\sum_{r=1}^d \left(\frac{a_r}{x_{0,\e_i+\e_r}^{k_r}}+\frac{b_r}{x_{0,\e_i-\e_r}^{l_r}}\right)=b_i \varepsilon^{-l_i}+O(1),\\
x_{1,-\e_i}&=a_i \varepsilon^{-k_i}+O(1),\\
x_{2,\pm \e_i \pm \e_j}&= -x_{0,\pm \e_i\pm \e_j}+\sum_{r=1}^d \left(\frac{a_r}{x_{1,\pm \e_i\pm \e_j+\e_r}^{k_r}}+\frac{b_r}{x_{1,\pm \e_i\pm \e_j-\e_r}^{l_r}}\right)=O(1)\ (i\neq j),\\
x_{2,\pm 2\e_i}&= -x_{0,\pm 2\e_i}+\sum_{r=1}^d \left(\frac{a_r}{x_{1,\pm 2\e_i + \e_r}^{k_r}}+\frac{b_r}{x_{1,\pm 2\e_i-\e_r}^{l_r}}\right)=O(1),\\
x_{2,\b0}&=-x_{0,\b0}+\sum_{r=1}^d \left(\frac{a_r}{x_{1,\e_r}^{k_r}}+\frac{b_r}{x_{1,-\e_r}^{l_r}}\right)=-\varepsilon+\sum_{r=1}^d \left(\frac{a_r}{(b_r \varepsilon^{-l_r}+O(1))^{k_r}}+\frac{b_r}{(a_r \varepsilon^{-k_r}+O(1))^{l_r}}\right)\\
&=-\varepsilon+\sum_{r=1}^d \varepsilon^{k_r l_r} \left( a_r b_r^{-k_r} +b_r a_r^{-l_r} + O(\varepsilon^{\min[k_r, l_r]}) \right)=-\varepsilon+O\left( \varepsilon^{\min [k_r l_r]_{1\le r\le d}} \right),\\
&=-\varepsilon+O(\varepsilon^M),
\end{align*}
where $M:=\min_{1\le r\le d}[k_r l_r]$.
In the next step we observe the essential cancellations as follows:
\begin{align}
x_{3,\e_i}&=-x_{1,\e_i}+\sum_{r=1}^d \left(\frac{a_r}{x_{2,\e_i+\e_r}^{k_r}}+\frac{b_r}{x_{2,\e_i-\e_r}^{l_r}}\right) \notag\\
&=-b_i \varepsilon^{-l_i}+\frac{b_i}{\left( -\varepsilon + O(\varepsilon^{M}) \right)^{l_i}}+O(1) \notag\\
&=-b_i \varepsilon^{-l_i}+(-1)^{l_i} b_i \frac{1}{\varepsilon^{l_i} \left( 1+O(\varepsilon^{M-1}) \right)^{l_i}}+O(1) \label{eq:liparity} \\
&=O(\varepsilon^{M-l_i-1})+O(1) \label{eq:Mcondition} \\
&=O(1).\notag 
\end{align}
Here we have used the parity of $l_i$ in \eqref{eq:liparity} and the condition \eqref{condition} in \eqref{eq:Mcondition}, respectively.
Similarly we have
\[
x_{3,-\e_i}=O(1).
\]
It is clear from the above expansions that $x_{4,\b0}=O(1)$.
Let us remark that all the above $O(1)$ depend on the initial values other than $x_{0,\b0}$.
\end{proof}

From here on let us employ the initial condition $I$ above and prove that each iterate $x_{t,\n}$ $(t\ge 0)$ is
a Laurent polynomial of $I$  with coefficients in $\Z[\{a_i,b_i\}_{1\le i\le d}]$ and is irreducible as the Laurent polynomial, under the condition \eqref{condition} on the parameters $l_i,k_i$.
Let us assume that \eqref{pedKdV_eq} is expressed using a new variable $f_{t,\n}$ as
\begin{equation}\label{xtof}
x_{t,\n}=\frac{f_{t,\n}f_{t-2,\n}}{F_{t-1,\n}},\qquad F_{t,\n}:=\prod_{i=1}^d f_{t,\n+\e_i}^{k_i}f_{t,\n-\e_i}^{l_i}.
\end{equation}
Note that the above transformation between $x_{t,\n}$ and $f_{t,\n}$ originates in the singularity pattern in Proposition \ref{prop:singpattern}.
Then we have
\[
\frac{f_{t+2,\n}f_{t,\n}}{F_{t+1,\n}}+\frac{f_{t,\n}f_{t-2,\n}}{F_{t-1,\n}}
=\sum_{i=1}^d \left\{ 
\frac{a_iF_{t,\n+\e_i}^{k_i}}{f_{t+1,\n+\e_i}^{k_i}f_{t-1,\n+\e_i}^{k_i}}
+\frac{b_iF_{t,\n-\e_i}^{l_i}}{f_{t+1,\n-\e_i}^{l_i}f_{t-1,\n-\e_i}^{l_i}}
\right\}.
\]
Thus \eqref{pedKdV_eq} is transformed into the following generalized tau-function form:
\begin{equation}\label{poly-form}
f_{t+2,\n}=-\frac{F_{t+1,\n}f_{t-2,\n}}{F_{t-1,\n}}
+\sum_{i=1}^d \left\{ 
\frac{a_iF_{t+1,\n}F_{t,\n+\e_i}^{k_i}}{f_{t,\n}f_{t+1,\n+\e_i}^{k_i}f_{t-1,\n+\e_i}^{k_i}}
+\frac{b_iF_{t+1,\n}F_{t,\n-\e_i}^{l_i}}{f_{t,\n}f_{t+1,\n-\e_i}^{l_i}f_{t-1,\n-\e_i}^{l_i}}
\right\},
\end{equation}
whose initial variables are $f_{t,\n}$ $(t=-4,-3,-2,-1)$ regarding the correspondence with $x_{t,\n}$ $(t=-2,-1)$.
We define
\begin{align*}
F_{t,\n}^{(+i)}&:=\frac{F_{t,\n}}{f_{t,\n+\e_i}^{k_i}},\quad 
F_{t,\n}^{(-i)}:=\frac{F_{t,\n}}{f_{t,\n-\e_i}^{l_i}},\\
G^{(+i)}_{t,\n}&:=f_{t,\n}^{k_il_i-1}F_{t+1,\n}^{(+i)}F_{t-1,\n}^{(+i)}\left(F_{t,\n+\e_i}^{(-i)}\right)^{k_i},
G^{(-i)}_{t,\n}:=f_{t,\n}^{k_il_i-1}F_{t+1,\n}^{(-i)}F_{t-1,\n}^{(-i)}\left(F_{t,\n-\e_i}^{(+i)}\right)^{l_i}.
\end{align*}
Then \eqref{poly-form} is equivalent to
\begin{align}
f_{t+2,\n}&=\frac{1}{F_{t-1,\n}} \left[-F_{t+1,\n}f_{t-2,\n}
+\sum_{i=1}^d f_{t,\n}^{k_il_i-1}\left\{ 
a_iF_{t+1,\n}^{(+i)}F_{t-1,\n}^{(+i)}\left(F_{t,\n+\e_i}^{(-i)}\right)^{k_i}
\right.\right. \notag \\
&\qquad\qquad \qquad  \qquad \left.\left. 
+b_iF_{t+1,\n}^{(-i)}F_{t-1,\n}^{(-i)}\left(F_{t,\n-\e_i}^{(+i)}\right)^{l_i}
\right\}\right]
\label{poly-form2} \\
&=\frac{1}{F_{t-1,\n}} \left[-F_{t+1,\n}f_{t-2,\n}
+\sum_{i=1}^d \left\{a_i G^{(+i)}_{t,\n}+b_iG^{(-i)}_{t,\n}\right\} \right].
\label{poly-abb}
\end{align}

It is worth noting that the two-dimensional equation \eqref{eq:nonlinear} (resp. \eqref{eq:laurent}) in the introduction is obtained by taking $d=1$ and $s:=\frac{1}{2}(t+n_1)$, $m:=\frac{1}{2}(t-n_1)$ of \eqref{pedKdV_eq} (resp. \eqref{poly-form}).
Since the equation \eqref{eq:nonlinear} is nonintegrable \cite{qintegrable}, the equation \eqref{pedKdV_eq}
is also nonintegrable in the sense of degree growth.

\subsection{Laurent, the irreducibility and the coprimeness properties of \eqref{pedKdV_eq} and \eqref{poly-form}}

Let us denote by $\mR$ the following ring of Laurent polynomials:
\begin{equation}\label{Laurent_ring_mR}
\mR:=\Z\left[\{f_{-4,\n}^\pm,f_{-3,\n}^\pm,f_{-2,\n}^\pm,f_{-1,\n}^\pm\},\, \{a_i,b_i\}\right].
\end{equation}
Here is our main theorem:
\begin{theorem}\label{th-exdKdV_f}
Let us assume the condition \eqref{condition}.
Then in equation \eqref{poly-form}, we have that $f_{t,\n} \in \mR$
for arbitrary $t,\,\n$, and that $f_{t,\n}$ is irreducible in $\mR$. Moreover, $f_{t,\n}$ and $f_{s,\m}$ are coprime with each other in $\mR$ for arbitrary $(t,\,\n)\neq (s,\,\m)$.
\end{theorem}
\begin{remark}
Since all the values of $k_i,l_i$ are even, the condition \eqref{condition} is equivalent to $\min_i[l_ik_i-1]\ge \max_i[l_i,\, k_i]$.
Also note that the condition \eqref{condition} is trivially satisfied if $k_i=l_i=k\in 2\Z_{>0}$ for all $i$.
\end{remark}
The proof of Theorem \eqref{th-exdKdV_f} is quite laborous and will be given in \textsection \ref{subsec:prf}.
Using Theorem \eqref{th-exdKdV_f} we obtain the coprimeness of the iterates of \eqref{pedKdV_eq}:
\begin{theorem}\label{ex-dKdV_coprime}
Suppose that $|t-t'|+\sum_{i=1}^d |n_i-n_i'|>2$, where $\n=(n_1,\cdots,n_d), \n'=(n_1',\cdots, n_d')$.
Then the two iterates $x_{t,\n},\,x_{t',\n'}$ of the equation \eqref{pedKdV_eq}
are coprime with each other in the field
\[
\mathbb{Q}\left(\{x_{-2,\m}, x_{-1,\m}\}_{\m\in\mathbb{Z}^d},\, \{a_i,b_i\}_{1\le i\le d}\right).
\]
\end{theorem}
\begin{proof}
Let us consider the equation \eqref{poly-form} with the following initial values:
\begin{align*}
f_{-4,\n}&:=g_{-4,\n}, \ f_{-3,\n}:=g_{-3,\n},\\
f_{-2,\n}&:=\frac{x_{-2,\n}F_{-3,\n}}{g_{-4,\n}}\ \left(=\frac{x_{-2,\n}\prod_{i=1}^d g_{-3,\n+\e_i}^{k_i}g_{-3,\n-\e_i}^{l_i}}{g_{-4,\n}}\right),\\
f_{-1,\n}&:=\frac{x_{-1,\n}F_{-2,\n}}{g_{-3,\n}} \left(=\frac{x_{-1,\n}\prod_{i=1}^d f_{-2,\n+\e_i}^{k_i}f_{-2,\n-\e_i}^{l_i}}{g_{-3,\n}}\right),
\end{align*}
where $g_{-3,\n}$, $g_{-4,\n}$ are auxiliary variables.
From Theorem \ref{th-exdKdV_f} we have
\[
f_{t,\n} \in  \Z\left[\{f_{-4,\m}^\pm,f_{-3,\m}^\pm,f_{-2,\m}^{\pm},f_{-1,\m}^{\pm}\},\,\{a_i,b_i\} \right]=\Z\left[\{g_{-4,\m}^\pm,g_{-3,\m}^\pm,x_{-2,\m}^{\pm},x_{-1,\m}^{\pm}\},\,\{a_i,b_i\} \right],
\]
and $f_{t,\n}$ is irreducible.
Since $x_{t,\n}$ is independent of $\{g_{-4,\m}, g_{-3,\m}\}$, any pair $x_{t,\n}$, $x_{t',\n'}$ is coprime in $\mathbb{Q}(\{x_{-2,\m}, x_{-1,\m}\},\{a_i,b_i\})$.
\end{proof}

\subsection{Proof of Theorem \eqref{th-exdKdV_f}} \label{subsec:prf}
The proof is given inductively with respect to $t (\ge 0)$.

\setcounter{step}{0}
\begin{step}[case of $t=0$]
Taking $t=-2,-1,0$ in \eqref{poly-form2} we immediately obtain $f_{t,\n}\in \mR$ for $t=0,1,2$.
In particular $f_{0,\n}$ is linear with respect to $f_{-4,\n}$ and is not a unit.
Thus $f_{0,\n}$ is irreducible for any $\n$.
\end{step}

\begin{step}[case of $t=1$]
From a property of the factorization of Laurent polynomials in Lemma \ref{lemma-common}, we have 
\[
f_{1,\n}=\left(\prod_{\m} f_{0,\m}^{\alpha_{\m}}\right) f_{1,\n}'\quad \left(\alpha_m \in \Z_{\ge 0}\right),
\]
where $f_{1,\n}'$ is irreducible.
If we fix a certain $\n$, then, from the time evolution rule, $\alpha_{\m}$ must be zero unless $\n=\m$ or $\n=\m\pm \e_i$ for some $i$.
It is sufficient to prove that $\alpha_{\n}=\alpha_{\n\pm\e_i}=0$ for any $i$.
Let us substitute the following values in the initial variables:
\begin{align}
&f_{-4,\m}=-1+\delta_{\m,\n}+\sum_{i=1}^d(a_i+b_i)\;\;(\forall\m),\quad f_{-3,\m}=f_{-2,\m}=f_{-1,\m}=1 \ (\forall \m), \label{initt1}
\end{align}
then we have $f_{0,\m}=-f_{-4,\m}+\sum_{i=1}^d (a_i+b_i)=1-\delta_{\m,\n}$, where $\delta_{\m,\n}$ is Kronecker's delta.
Thus
\[
f_{1,\m}=\left\{
\begin{array}{cl}
-1+\sum_{i=1}^d (a_i+b_i) &\quad (\m=\n) \\
  b_i &\quad (\m=\n+\e_i) \\
a_i &\quad (\m=\n-\e_i) \\
-1+\sum_{i=1}^d (a_i+b_i) &\quad (\mbox{otherwise}) 
\end{array}
\right.
.
\]
Therefore $f_{1,\n}\neq 0$ and $f_{0,\n}=0$ at the same time, which indicates that $\alpha_{\n}=0$.
A similar discussion leads to $\alpha_{\n \pm \e_i}=0$ and thus $f_{1,\n}$ is irreducible.
The pairwise coprimeness is clear from the fact that $f_{1,\n}$ has the variable $f_{-4,\n'}$ for at least one $\n'$, which is absent from $f_{1,\m}$ with $\n \neq \m$.
\end{step}

\begin{step}[case of $t=2$]
We have
\[
f_{2,\n}=\left(\prod_{\m} f_{0,\m}^{\alpha_{\m}}\right)f_{2,\n}'\quad (\alpha_{\m} \in \Z_{\ge 0}),
\]
where $f_{2,\n}'$ is irreducible.
We need to prove that the indices $\alpha_{\m}$ are all zero.
Let us take the initial values \eqref{initt1} as in the case of $t=1$.
Since $l_i$ is even,  the iterate $F_{0,\n+\e_i}$ must have a factor $f_{0,\n}^2$.
Thus using \eqref{poly-form2} we have
$f_{2,\n}=-F_{1,\n}=-\prod_{i=1}^d (a_ib_i) \ne 0$.
To prove $f_{2,\n\pm\e_i\pm\e_j}\neq 0$ (excluding the case above), it is sufficient to
assume that $a_i,\,b_i$ are negative constants. Then we observe that every term in the right hand side of \eqref{poly-form2} is nonpositive, and in particular the first term is negative: i.e., $-\frac{F_{1,\n}f_{-2,\n}}{F_{-1,\n}}<0$ since $F_{1,\n}>0$.
Finally, it is readily obtained that $f_{2,\m}\neq 0$ for every $\m \ne  \n_0 \pm \e_i \pm \e_j$,
since by taking $a_i=b_i=0$ we have $f_{2,\n}=-F_{1,\n}$.
Thus the irreducibility of $f_{2,\n}$ is proved.
The coprimeness of $f_{2,\n}$ and $f_{2,\m}$ for $\n\neq \m$ is readily obtained.
The coprimeness of $f_{2,\n}$ and $f_{1,\m}$ is proved later in a more general setting
in Step~\ref{mutually_coprime_prop}.
\end{step}

\begin{step}[case of $t=3$]
The proof of the Laurent property requires some tedious calculations.
We try to elaborate on this to see how the condition \eqref{condition} is used.
To ease notation let us define $g_{3,\n}$ by
\[
g_{3,\n}:=F_{0,\n}f_{3,\n} =-F_{2,\n}f_{-1,\n}+\sum_{i=1}^d \left\{a_i G^{(+i)}_{1,\n}+b_iG^{(-i)}_{1,\n}\right\}.
\]
Then we have
\begin{equation}
g_{3,\n}=-F_{2,\n}f_{-1,\n}+a_1G_{1,\n}^{(+1)}+f_{0,\n+\e_1}^{k_1}\times (\mbox{polynomial in $f_{t,\m}$ ($t \le 2$)}) \label{g3expression}.
\end{equation}
Let us define  $g_{3,\n}^{(1)}:=-F_{2,\n}f_{-1,\n}+a_1G_{1,\n}^{(+1)}$
and calculate further:
since
$G_{1,\n}^{(+1)}=f_{1,\n}^{k_1l_1-1}F_{2,\n}^{(+1)}F_{0,\n}^{(+1)}\left(F_{1,\n+\e_1}^{(-1)}\right)^{k_1}$,
we have
\[
g_{3,\n}^{(1)}=F_{2,\n}^{(+1)}\left( -f_{2,\n+\e_1}^{k_1}f_{-1,\n}+a_1f_{1,\n}^{k_1l_1-1}F_{0,\n}^{(+1)}\left(F_{1,\n+\e_1}^{(-1)}\right)^{k_1}      \right).
\]
Moreover,
\begin{align*}
f_{2,\n+\e_1}&=\frac{1}{F_{-1,\n+\e_1}} \left[-F_{1,\n+\e_1}f_{-2,\n+\e_1}
+\sum_{i=1}^d \left\{a_i G^{(+i)}_{0,\n+\e_1}+b_iG^{(-i)}_{0,\n+\e_1}\right\} \right]\\
&=\frac{1}{F_{-1,\n+\e_1}} \left[-F_{1,\n+\e_1}f_{-2,\n+\e_1}+f_{0,\n+\e_1}^{\min_i[k_il_i-1]}\times (\mbox{polynomial in $f_{t,\m}$ ($t \le 1$)})\right].
\end{align*}
Thus using the condition \eqref{condition} we have
\begin{equation}
f_{2,\n+\e_1}^{k_1}=\frac{1}{F_{-1,\n+\e_1}^{k_1}}\left[F_{1,\n+\e_1}^{k_1}f_{-2,\n+\e_1}^{k_1}+f_{0,\n+\e_1}^{k_1}\times (\mbox{polynomial in $f_{t,\m}$ ($t \le 1$)})\right] \label{f2k1expression}.
\end{equation}
Therefore using $F_{1,\n+\e_1}=F_{1,\n+\e_1}^{(-1)}f_{1,\n}^{l_1}$, we obtain
\begin{align*}
g_{3,\n}^{(1)}&=
\frac{F_{2,\n}^{(+1)} }{F_{-1,\n+\e_1}^{k_1}}\Bigl[\left(F_{1,\n+\e_1}^{(-1)}\right)^{k_1}f_{1,\n}^{k_1l_1-1} (
-f_{1,\n}f_{-2,\n+\e_1}^{k_1}f_{-1,\n}+a_1F_{0,\n}^{(+1)}F_{-1,\n+\e_1}^{k_1}) \\
&\qquad\qquad \qquad +f_{0,\n+\e_1}^{k_1}\times (\mbox{polynomial in $f_{t,\m}$ ($t \le 1$)})
\Bigr].
\end{align*}
Furthermore, we have
\begin{align*}
f_{1,\n}&=\frac{1}{F_{-2,\n}} \left[-F_{0,\n}f_{-3,\n}
+\sum_{i=1}^d \left\{a_i G^{(+i)}_{-1,\n}+b_iG^{(-i)}_{-1,\n}\right\} \right]\\
&=\frac{1}{F_{-2,\n}} \left[a_1G^{(+1)}_{-1,\n}+f_{0,\n+\e_1}^{k_1}\times
(\mbox{polynomial in $f_{t,\m}$ ($t \le 0$)})\right].
\end{align*}
Therefore,
\[
-f_{1,\n}f_{-2,\n+\e_1}^{k_1}f_{-1,\n}+a_1F_{0,\n}^{(+1)}F_{-1,\n+\e_1}^{k_1}=\frac{f_{0,\n+\e_1}^{k_1}}{F_{-2,\n}}\times
(\mbox{polynomial in $f_{t,\m}$ ($t \le 0$)}).
\]
From the calculations above we have
\[
g_{3,\n}^{(1)}=\frac{f_{0,\n+\e_1}^{k_1}}{F_{-1,\n+\e_1}^{k_1}F_{-2,\n}} \left(\mbox{polynomial in $f_{t,\m}$ ($t \le 2$)}\right).
\]
Thus using the equation \eqref{g3expression} we have
\begin{equation*}
g_{3,\n}=\frac{f_{0,\n+\e_1}^{k_1}}{F_{-1,\n+\e_1}^{k_1}F_{-2,\n}} \left(\mbox{polynomial in $f_{t,\m}$ ($t \le 2$)}\right).
\end{equation*}
Conducting a  similar calculation as we have done for $k_1$ in \eqref{f2k1expression} for $k_2,\cdots, k_d,l_1,\cdots, l_d$ and using the condition \eqref{condition} repeatedly, we have
\[
g_{3,\n}=\frac{F_{0,\n}}{\prod_{i=1}^d\left(F_{-1,\n+\e_i}^{k_i}F_{-1,\n-\e_i}^{l_i}\right)F_{-2,\n}} \times \left(\mbox{polynomial in $f_{t,\m}$ ($t \le 2$)}\right). \label{g3expression2}
\]
Thus the denominator $F_{0,\n}$ of $f_{3, \n}=\frac{g_{3,\n}}{F_{0,\n}}$ is cancelled.
Thus we have proved that $f_{3,\n} \in \mR$.
\end{step}

\begin{step}[Laurent property from the coprimeness]
The following Claim \ref{ft2mR} is used to prove the Laurent and the irreducibility properties
inductively with respect to $t\ge 4$.
\begin{claim} \label{ft2mR}
Let us fix $t\ge 3$ and suppose that $f_{s,\n}\in\mR$ for every $s\le t+1$.  
Moreover let us suppose that $f_{t-1,\m_1}, \,f_{t-2,\m_2},\,f_{t-3,\m_3}$ are pairwise coprime for any $\m_1,\m_2,\m_3\in\mathbb{Z}^d$.
Then we have $f_{t+2,\n} \in \mR$.
\end{claim}
Claim \ref{ft2mR} is readily proved using the expression
\begin{equation}\label{main_cal}
F_{t-1,\n}f_{t+2,\n}=\frac{F_{t-1,\n}\times \left(\mbox{polynomial in $f_{s,\m}$ ($s \le t+1$)}\right)}{\prod_{i=1}^d\left(F_{t-2,\n+\e_i}^{k_i}F_{t-2,\n-\e_i}^{l_i}\right)F_{t-3,\n}}.
\end{equation}
\end{step}

\begin{step}[preparations on the degree growth]
From here on we study the irreducibility and the coprimeness of $f_{t,\n}$, since the Laurent property automatically follows from the two properties using Claim \ref{ft2mR}.
Let us define the sequence $\{y_t\}_{t\ge -4}$ by the values $y_t$ of $f_{t,\n}$ when we take particular initial values: all the initial variables are substituted by $1$ with the 
exception of $f_{-4,\n}=-x$ for every $\n$, where $x$ is an auxiliary variable.
Then $y_t$ satisfies the following recurrence:
\begin{equation}\label{y_eq}
y_{t+1}=-\frac{y_t^Ny_{t-3}}{y_{t-2}^N}+\sum_{i=1}^d \left\{ \frac{a_i y_t^{N-k_i}y_{t-1}^{k_iN-1}}{y_{t-2}^{k_i}}   +\frac{b_i y_t^{N-l_i}y_{t-1}^{l_iN-1}}{y_{t-2}^{l_i}}  \right\},
\end{equation}
where $N:=\sum_{i=1}^d (k_i+l_i)$, $y_{-4}=-x$, $y_{-3}=y_{-2}=y_{-1}=1$.
It is easy to see that $y_0=x+\sum_{i=1}^d (a_i+b_i)$.
\begin{claim}\label{deg_prop1}
In equation \eqref{y_eq},
$y_t$ $(t \ge 0)$ is a polynomial in $x$ whose constant term is nozero.
Moreover we have $y_{t-1}\,\big|\, y_t$ for every $t$.
\end{claim}
\begin{proof}
The properties are shown by induction.
Let us define $x_t$ by
\[
x_t=\frac{y_t y_{t-2}}{y_{t-1}^N}
\]
for $t\ge -2$.
Then $x_t$ satisfies the following recurrence
\[
 x_{-2}=-x,\;x_{-1}=1,\; x_{t+1}+x_{t-1}=\sum_{i=1}^d \left( \frac{a_i}{x_t^{k_i}}+\frac{b_i}{x_t^{l_i}} \right) \ (t\ge -1).
\]
Let us show that $x_t\neq 0$ for every $t$.
It is sufficient to prove that $x_t\neq 0$ under a special case $a_1=a$, $a_j=0$ ($j \ne 1$), $b_i=0$.
Let us define $\gamma_t$ as the degree of $x_t$ as a rational function of $a$.
Then $\gamma_{-2}=0, \gamma_{-1}=1, \gamma_{0}=k_1, \gamma_{1}=k_1^2+1$,
and therefore we have
\[
\gamma_{t+1} \ge k_1 \gamma_t-\gamma_{t-1}-1,
\]
from which we inductively show that $\gamma_{t+1} \ge \gamma_t+1$ ($t \ge 0$).
Thus $x_t\neq 0$, which indicates that $y_t \neq 0$ when $x=0$.
Since $y_t$ has a nonzero constant term, the denominators of the RHS of \eqref{y_eq} do not contain monomial factors by dividing by $y_{t-2}$.
Thus $y_{t+1}$ must be a polynomial.
\end{proof}
\begin{claim}\label{deg_prop2}
Let $d_t$ be the degree $\Ord_x(y_t)$ of $y_t$ with respect to $x$.
Then $d_t$ satisfies the following recurrence
\begin{equation}\label{degree_count}
d_{t+1}=Nd_t-Nd_{t-2}+d_{t-3} \qquad (t \ge 0),
\end{equation}
where $d_{-3}=d_{-2}=d_{-1}=0,\,d_0=1$.
\end{claim}
\begin{proof}
Note that $y_{-4}=x, y_{-3}=y_{-2}=y_{-1}=1, y_0=x+N$.
The statement is true for $t=0,\,1,\,2$ since $d_1=N,\,d_2=N^2,\,d_3=N^3-N$.
%
Since \eqref{degree_count} is equivalent to
\begin{equation}
d_{t+2}-Nd_{t+1}+d_t=d_t-Nd_{t-1}+d_{t-2}, \label{dtrecurrence}
\end{equation}
it is sufficient to prove the following:
\[
d_{t}-Nd_{t-1}+d_{t-2}=\left\{
\begin{array}{cl}
1 & \quad (t:\,\mbox{even})\\
0 & \quad (t:\,\mbox{odd})
\end{array}
\right.
.
\]
Proof is done by induction.
Let us assume the statement up to $t=s$ (for a fixed $s \ge 1$)  and study the case of
$t=s+1$: we have
\begin{align*}
\Ord_x\left(\frac{y_s^Ny_{s-3}}{y_{s-2}^N} \right)&=Nd_s-Nd_{s-2}+d_{s-3},\\
\Ord_x\left(\frac{ y_s^{N-k_i}y_{s-1}^{k_iN-1}}{y_{s-2}^{k_i}}  \right)&
= (N-k_i)d_s+(k_iN-1)d_{s-1}-k_id_{s-2} ,
\end{align*}
\begin{align*}
&Nd_s-Nd_{s-2}+d_{s-3}-\left\{  (N-k_i)d_s+(k_iN-1)d_{s-1}-k_id_{s-2} \right\}\\
&=k_i(d_s-Nd_{s-1}+d_{s-2})+(d_{s-1}-Nd_{s-2}+d_{s-3}) \ge 1,
\end{align*}
from the induction hypothesis.
Therefore we have
\[
\Ord_x\left(\frac{y_s^Ny_{s-3}}{y_{s-2}^N} \right)>\Ord_x\left(\frac{ y_s^{N-k_i}y_{s-1}^{k_iN-1}}{y_{s-2}^{k_i}}  \right).
\]
Similarly we have
\[
\Ord_x\left(\frac{y_s^Ny_{s-3}}{y_{s-2}^N} \right)>\Ord_x\left(\frac{ y_s^{N-l_i}y_{s-1}^{l_iN-1}}{y_{s-2}^{l_i}}  \right).
\]
Thus the first term in the RHS of \eqref{y_eq} has the largest degree.
Therefore we have
$d_{s+1}=Nd_s-Nd_{s-2}+d_{s-3}$.
\end{proof}
\end{step}

\begin{step}[coprimeness from the irreducibility]\label{mutually_coprime_prop}
We show that the two iterates $f_{t,\n}$ and $f_{s,\m}$ with $(t,\n)\neq (s,\m)$ are coprime with each other if they are both irreducible. 

From Claims \ref{deg_prop1} and \ref{deg_prop2}, it is shown that, if $f_{t,\n}$ and $f_{s, \m}$ ($s > t$) are both irreducible, they must be coprime with each other.
Let us suppose otherwise.
We have $y_s\, |\, y_t$ from Claim \ref{deg_prop1}. From the irreducibility  we have $y_t=x^\alpha y_s$ for some integer $\alpha$.
On the other hand, the constant term of $y_t$ is nonzero from Claim \ref{deg_prop1}, thus
$\alpha=0$.
Therefore we have $y_t=y_s$, which in turn contradicts Claim \ref{deg_prop2}.
Next in the case of $t=s$, $f_{t,\n}$ and $f_{t,\m}$ ($\n \ne \m$) are pairwise coprime if they are both irreducible, which is immediate from the fact that
they both have terms of $f_{-4,\n'}$ for some $\n'$ which cannot be cancelled out
by factoring some monomials.
\end{step}

\begin{step}[irreducibility for $t=3$]
The irreducibility for $t \le 2$ is already proved.
Now let us study the case of $t = 3$. We have the factorization
\[
f_{3,\n}=\left(\prod_{\m\in\mathbb{Z}^d} f_{0,\m}^{\alpha_{\m}} \right)f_{3,\n}'\quad (\alpha_{\m} \in \Z_{\ge 0}),
\]
where $f_{3,\n}'$ is irreducible.
In a similar manner to the discussion in previous paragraphs, let us prove $\alpha_{\m}=0$ for all $\m\in\mathbb{Z}^d$.
From the evolution equation it is sufficient to prove $\alpha_{\m} \ne 0$ only when
\[
\n=\m\pm\e_i\pm\e_j\pm\e_p\ (i,j,p\in\{1,2,\cdots,d\}).
\]
We can take $\m=\b0$ without losing generality.
Let us take the initial values $f_{-4,\b0}=-x+\sum_{i=1}^{d}(a_i+b_i)$, $f_{-4,\m}=-1+\sum_{i=1}^{d}(a_i+b_i)$ ($\m \ne \b0$),
and $f_{-3,\m}=f_{-2,\m}=f_{-1,\m}=1$ for all $\m$.
Then we have $f_{0,\b0}=x, f_{0,\m}=1 \;\;(m \ne \b0)$, and
\[
f_{1,\n}=-F_{0,\n}+\sum_{i=1}^d (a_iG_{-1,\n}^{(+i)}+b_iG_{-1,\n}^{(-i)})=-F_{0,\n}+\sum_{i=1}^d( a_iF_{0,\n}^{(+i)}+b_iF_{0,\n}^{(-i)}),
\]
where
\[
F_{0,\e_i}=x^{l_i},\; F_{0,-\e_i}=x^{k_i},\; F_{0,\e_i}^{(-i)}=1,\; F_{0,-\e_i}^{(+i)}=1,\; F_{0,\pm\e_i}^{(\pm j)}=F_{0,\pm \e_i}
\;(\mbox{otherwise}).
\]
Thus
\[
f_{1,\e_i}=b_i +(A-b_i)x^{l_i},\quad f_{1,-\e_i}=a_i+(A-a_i)x^{k_i}, 
\quad f_{1, \m}=A \;\;(\m \ne \pm\e_i),
\]
where we have defined $A:=\sum_{i=1}^d(a_i+b_i)-1$.
Next let us compute
\[
f_{2,\n}=-F_{1,\n}+\sum_{i=1}^d\left\{ \frac{a_iF_{1,\n}F_{0,\n+\e_i}^{k_i}}{f_{0,\n}f_{1,\n+\e_i}^{k_i}}
+\frac{b_i F_{1,\n}F_{0,\n-\e_i}^{l_i}}{f_{0,\n}f_{1,\n-\e_i}^{l_i}}\right\}.
\]
First, we have $F_{1,\m}=A^N$ for $\m\neq \pm \e_i \pm \e_j$ $(1\le i,j\le d)$
and in other cases $F_{1,\m}$ is a polynomial in $x$ as
\begin{align*}
F_{1,\b0}&=\prod_{i=1}^d(b_i +(A-b_i)x^{l_i})^{k_i}(a_i+(A-a_i)x^{k_i})^{l_i},\\
F_{1,2\e_i}&=A^{N-l_i} (b_i+(A-b_i)x^{l_i})^{l_i},\ F_{1,-2\e_i}=A^{N-k_i} (a_i+(A-a_i)x^{k_i})^{k_i},\\
F_{1,\e_i+\e_j}&=A^{N-l_i-l_j} (b_i+(A-b_i)x^{l_i})^{l_j}(b_j+(A-b_j)x^{l_j})^{l_i}\ (i\neq j),\\
F_{1,-\e_i-\e_j}&=A^{N-k_i-k_j} (a_i+(A-a_i)x^{k_i})^{k_j} (a_j+(A-a_j)x^{k_j})^{k_i}\ (i\neq j),\\
F_{1,\e_i-\e_j}&=A^{N-l_i-k_j} (b_i+(A-b_i)x^{l_i})^{k_j}(a_j+(A-a_j)x^{k_j})^{l_i}\ (i\neq j),
\end{align*}
where we have used $N:=\sum_{i=1}^d (k_i+l_i)$.
From here on let us substitute $A=1$ and use the notations $a_i(x):=a_i+(1-a_i)x^{k_i}$, $b_i(x):=b_i+(1-b_i)x^{l_i}$, since
if we prove $f_{3,\m}\neq 0$ for a particular $A$, then it is trivial that $f_{3,\m}\neq 0$ for an indeterminant $A$.
When $A=1$, we have $f_{1,\e_i}=b_i(x), f_{1,-\e_i}=a_i(x), f_{1,\m}=1 \; (\m \ne \pm \e_i)$
and
\begin{align*}
&F_{1,\b0}=\prod_{i=1}^d b_i (x)^{k_i}a_i(x)^{l_i},\;
F_{1,2\e_i}=b_i(x)^{l_i},\;
F_{1,-2\e_i}=a_i(x)^{k_i},\;
F_{1,\e_i+\e_j}=b_i(x)^{l_j}b_j(x)^{l_i},\\
&F_{1,-\e_i-\e_j}= a_i(x)^{k_j} a_j(x)^{k_i},\;
F_{1,\e_i-\e_j}=b_i(x)^{k_j}a_j(x)^{l_i},\;
F_{1,\m}=1 \;(\mbox{otherwise}).
\end{align*}
Then we obtain $f_{2,\m}=1$ for $\m\neq \pm \e_i \pm \e_j$ $(1\le i,j\le d)$
and in other cases we have
\begin{align*}
f_{2,\b0}&=\left(\prod_{i=1}^db_i (x)^{k_i}a_i(x)^{l_i}\right) \left\{-1+\sum_{i=1}^d \left( \frac{a_ix^{k_il_i-1}}{b_i(x)^{k_i}}
+  \frac{b_i x^{k_il_i-1}}{a_i(x)^{l_i}} \right)\right\},\\
f_{2,2\e_i}&=b_i(x)^{l_i}\left\{1+b_i\left(\frac{x^{l_i^2}}{b_i(x)^{l_i}}-1\right)       \right\},\ 
f_{2,-2\e_i}=a_i(x)^{k_i}\left\{1+a_i\left(\frac{x^{k_i^2}}{a_i(x)^{k_i}}-1\right)       \right\},\\
f_{2,\e_i+\e_j}&=b_i(x)^{l_j}b_j(x)^{l_i}\left\{1+ b_i\left(\frac{x^{l_il_j}}{(b_j (x))^{l_i}}-1\right)      +b_j\left(\frac{x^{l_il_j}}{(b_i (x))^{l_j}}-1\right)      \right\}\ (i\neq j),\\
f_{2,-\e_i-\e_j}&= a_i(x)^{k_j} a_j(x)^{k_i}\left\{1+ a_i\left(\frac{x^{k_ik_j}}{a_j(x)^{k_i}}-1\right)  +a_j\left(\frac{x^{k_ik_j}}{(a_i (x))^{k_j}}-1\right)                                \right\}\ (i\neq j),\\
f_{2,\e_i-\e_j}&=b_i(x)^{k_j}a_j(x)^{l_i}\left\{1+ b_i\left(\frac{x^{l_ik_j}}{a_j(x)^{l_i}}-1\right)      +a_j\left(\frac{x^{l_ik_j}}{(b_i(x))^{k_j}}-1\right)      \right\}\ (i\neq j).
\end{align*}
%
%
%
%
%

Finally let us compute the value of $f_{3,\n}$ when we substitute the values above in
the initial variables.
Recall that \[
f_{3,\n}=-\frac{F_{2,\n}}{F_{0,\n}}+\sum_{i=1}^d \left\{\frac{a_iF_{2,\n}F_{1,\n+\e_i}^{k_i}}{f_{1,\n}f_{2,\n+\e_i}^{k_i}f_{0,\n+\e_i}^{k_i}}+\frac{b_iF_{2,\n}F_{1,\n-\e_i}^{l_i}}{f_{1,\n}f_{2,\n-\e_i}^{l_i}f_{0,\n-\e_i}^{l_i}}\right\}.
\]
We shall prove that $f_{3,\n}|_{x\to 0} \neq 0$, from which $\alpha_{\boldsymbol{n}}=0$ is obtained.
Let us study the cases of $\n=\pm \e_i\pm \e_j \pm \e_p$ one by one.

\begin{itemize}
\item The case of $\n=\pm\e_i$:

It is sufficient to prove the case of $\n=\e_1$.
Since
\[
F_{2,\e_1}=f_{2,\b0}f_{2,2\e_1}\prod_{i=2}^d f_{2,\e_1+\e_i}f_{2,\e_1-\e_i},
\]
we have
\begin{equation}
f_{3,\e_1}=F_{2,\e_1}\left[ -\frac{1}{x^{l_1}}+\sum_{i=1}^d    
\left\{\frac{a_iF_{1,\e_1+\e_i}^{k_i}}{f_{1,\e_1}f_{2,\e_1+\e_i}^{k_i}f_{0,\e_1+\e_i}^{k_i}}+\frac{b_iF_{1,\e_1-\e_i}^{l_i}}{f_{1,\e_1}f_{2,\e_1-\e_i}^{l_i}f_{0,\e_1-\e_i}^{l_i}}\right\}   \right] \label{f3e1expression}.
\end{equation}
Let us prove that \eqref{f3e1expression} converges to a non-zero value for $x\to 0$.
By taking the second term inside $\sum_{i=1}^d\{\,\}$ when $i=1$, we observe that the divergent term $x^{-l_1}$ is cancelled out:
\begin{align}
&-\frac{1}{x^{l_1}}+\frac{b_1F_{1,\b0}^{l_1}}{f_{1,\e_1}f_{2,\b0}^{l_1}f_{0,\b0}^{l_1}}=\frac{1}{x^{l_1}}\left( -1+ \frac{b_1F_{1,\b0}^{l_1}}{b_1(x)f_{2,\b0}^{l_1}}\right)\notag\\
&=\frac{1}{x^{l_1}}\left[ \frac{b_1-b_1(x)\left\{-1+\sum_{i=1}^d \left( \frac{a_ix^{k_il_i-1}}{b_i(x)^{k_i}}
+  \frac{b_i x^{k_il_i-1}}{a_i(x)^{l_i}} \right)\right\}^{l_1}}{b_1(x)\left\{-1+\sum_{i=1}^d \left( \frac{a_ix^{k_il_i-1}}{b_i(x)^{k_i}}
+  \frac{b_i x^{k_il_i-1}}{a_i(x)^{l_i}} \right)\right\}^{l_1}}\right]
\label{mideq_1}.
\end{align}
From $b_1(x)=b_1+(1-b_1)x^{l_1}$ and \eqref{condition}, the RHS of
\eqref{mideq_1} converges to $(b_1-1)/b_1$ when $x\to 0$．
All the other terms in \eqref{f3e1expression} converges to finite values:
\begin{align*}
&f_{1,\e_i}\, \rightarrow \, b_i, \;
f_{1,-\e_i}\, \rightarrow \, a_i, \;
f_{1,\m} = 1 \;\; (\m \ne \pm \e_i), \\
&F_{1,2\e_i}\,\rightarrow \, b_i^{l_i},\;
F_{1,-2\e_i}\,\rightarrow \, a_i^{k_i},\;
F_{1,\e_i+\e_j}\,\rightarrow \, b_j^{l_i}b_i^{l_j},\;
F_{1,-\e_i-\e_j}\,\rightarrow \, a_j^{k_i}a_i^{k_j},\;
F_{1,\e_i-\e_j}\,\rightarrow \, a_j^{l_i}b_i^{k_j},\\
&f_{2,\b0}\, \rightarrow \, -\prod_{i=1}^d b_i^{k_i}a_i^{l_i}, \;
f_{2,2\e_i} \,\rightarrow \,b_i^{l_i}\left(1-b_i\right),  \;
f_{2,-2\e_i} \, \rightarrow \,a_i^{k_i}\left(1-a_i\right), \\
&f_{2,\e_i+\e_j}\, \rightarrow \,b_i^{l_j}b_j^{l_i}\left(1- b_i-b_j\right), \;
f_{2,-\e_i-\e_j} \,\rightarrow \, a_i^{k_j} a_j^{k_i}\left(1- a_i-a_j\right),\;
f_{2,\e_i-\e_j} \, \rightarrow \,b_i^{k_j}a_j^{l_i}\left(1- b_i-a_j\right),\\
&f_{2,\m}=1\qquad (\mbox{otherwise}).
\end{align*}
Therefore $\overline{F}_{2,\e_1}:=F_{2,\e_1}\big|_{x\to 0} \neq 0$ (from here on we apply an overline to the variables to denote the substitution of $x\to 0$: e.g., $\overline{F}_*=F_{*} |_{x\to 0}$), and we have
\[
\overline{f}_{3,\e_1}=\overline{F}_{2,\e_1}\left[ \frac{b_1-1}{b_1}+\frac{a_1}{b_1(1-b_1)^{k_1}}+\sum_{i=2}^d \left( \frac{a_i}{b_1(1-b_1-b_i)^{k_i}}+\frac{b_i}{b_1(1-b_1-a_i)^{l_i}}   \right)  \right] \neq 0.
\]
\item The case of $\n=\pm 3\e_i$: it is sufficient to investigate the case of $\n=3\e_1$.
Since
\begin{align*}
f_{3,3\e_1}&=-\frac{F_{2,3\e_1}}{F_{0,3\e_1}}+\sum_{i=1}^d \left\{\frac{a_iF_{2,3\e_1}F_{1,3\e_1+\e_i}^{k_i}}{f_{1,3\e_1}f_{2,3\e_1+\e_i}^{k_i}f_{0,3\e_1+\e_i}^{k_i}}+\frac{b_iF_{2,3\e_1}F_{1,3\e_1-\e_i}^{l_i}}{f_{1,3\e_1}f_{2,3\e_1-\e_i}^{l_i}f_{0,3\e_1-\e_i}^{l_i}}\right\}\\
&=F_{2,3\e_1}\left[\frac{b_1F_{1,2\e_1}^{l_1}}{f_{2,2\e_1}^{l_1}} +A-b_1    \right],
\end{align*}
by taking $A=1$ we have
\[
\of_{3,3\e_1}=\oF_{2,3\e_1}\left( 1-b_1+\frac{b_1}{(1-b_1)^{l_1}}\right) \ne 0.
\]

\item The case of $\n=\pm 2\e_i \pm \e_j$ ($i \ne j$): it is sufficient to prove the case of $\n=2\e_1+\e_2$.
Since
\begin{align*}
f_{3,2\e_1+\e_2}&=-\frac{F_{2,2\e_1+\e_2}}{F_{0,2\e_1+\e_2}}+\sum_{i=1}^d \left\{\frac{a_iF_{2,2\e_1+\e_2}F_{1,2\e_1+\e_2+\e_i}^{k_i}}{f_{1,2\e_1+\e_2}f_{2,2\e_1+\e_2+\e_i}^{k_i}f_{0,2\e_1+\e_2+\e_i}^{k_i}}+\frac{b_iF_{2,2\e_1+\e_2}F_{1,2\e_1+\e_2-\e_i}^{l_i}}{f_{1,2\e_1+\e_2}f_{2,2\e_1+\e_2-\e_i}^{l_i}f_{0,2\e_1+\e_2-\e_i}^{l_i}}\right\}\\
&=F_{2,2\e_1+\e_2}\left[ A-b_1-b_2+\frac{b_1F_{1,\e_1+\e_2}^{l_1}}{f_{2,\e_1+\e_2}^{l_1}}+\frac{b_2F_{1,2\e_1}^{l_2}}{f_{2,2\e_1}^{l_2}}\right],
\end{align*}
we have
\[
\of_{3,2\e_1+\e_2}=\oF_{2,2\e_1+\e_2}\left( 1-b_1-b_2+\frac{b_1}{(1-b_1-b_2)^{l_1}}
+\frac{b_2}{(1-b_1)^{l_2}} \right) \ne 0.
\]

\item The case of $\n=\pm\e_i \pm \e_j \pm \e_p$ ($i \ne j \ne p \ne i$):
A direct calculation shows that
\[
\of_{3,\e_1+\e_2+\e_3}=\oF_{2,\e_1+\e_2+\e_3}\left( 1-b_1-b_2-b_3+\frac{b_1}{(1-b_2-b_3)^{l_1}}
+\frac{b_2}{(1-b_1-b_3)^{l_2}}+\frac{b_3}{(1-b_1-b_2)^{l_3}}  \right) \ne 0.
\]
\end{itemize}
\end{step}

\begin{step}[irreducibility for $t=4$]
Let us take the same initial condition as in the case of $t=3$ and let us assume that $A=1$.
From a discussion similar to the previous cases we have only to prove that $\of_{4,\n} \ne 0$.
Moreover, it is sufficient to prove $\of_{4,\n} \ne 0$ in the case of $\n=\pm \e_i \pm \e_j \pm \e_p \pm \e_q$.

In the case of $\n=\b0$ we have
\[
\of_{4,\b0}=\sum_{i=1}^d \left\{ \frac{a_i\oF_{3,\b0}\oF_{2,\e_i}^{k_i}}{\of_{2,\b0}\of_{3,\e_i}^{k_i}\of_{1,\e_i}^{k_i}} +\frac{b_i\oF_{3,\b0}\oF_{2,-\e_i}^{l_i}}{\of_{2,\b0}\of_{3,-\e_i}^{l_i}\of_{1,-\e_i}^{l_i}}\right\},
\]
each term of which is negative for indeterminates $a_i,\,b_i >0$ since $\of_{2,\b0}<0$.
Therefore $\of_{4,\b0} \ne 0$.
In the case of $\n \ne \b0$ we have
\[
\of_{4,\n}=\frac{\oF_{3,\n}}{\of_{2,\n}} \left[-\frac{\of_{2,\n}}{\oF_{1,\n}}
+\sum_{i=1}^d\left\{ \frac{a_i\oF_{2,\n+\e_i}^{k_i}}{\of_{3,\n+\e_i}^{k_i}\of_{1,\n+\e_i}^{k_i}}
+ \frac{b_i \oF_{2,\n-\e_i}^{l_i}}{\of_{3,\n-\e_i}^{l_i}\of_{1,\n-\e_i}^{l_i}}\right\}
 \right].
\]
\begin{itemize}
\item In the case of $\n=\pm 2\e_i$, it is sufficient to investigate the case of $\n=2\e_1$.
In this case
\[
\frac{\of_{2,2\e_1}}{\oF_{1,2\e_1}}=1-b_1,
\]
which does not depend on $\{a_i\}$.
Calculating further we have that all the terms in $\sum_{i=1}^d$ is dependent on $\{a_i\}$ and is not cancelled out by $1-b_1$: for example, 
\[
\frac{a_i\oF_{2,2\e_1+\e_i}^{k_i}}{\of_{3,2\e_1+\e_i}^{k_i}\of_{1,2\e_1+\e_i}^{k_i}}=\frac{a_i}{\left( 1-b_1-b_i+\frac{b_1}{(1-b_1-b_i)^{l_1}}
+\frac{b_i}{(1-b_1)^{l_i}}  \right)^{k_i}}\quad (i\neq 1).
\]
Therefore $\of_{4,2\e_1}$ cannot be zero.
\item In the case of $\n=\pm \e_i \pm \e_j$, it is sufficient to prove the case of $\n=\e_1+\e_2$. We have
\[
\frac{\of_{2,\e_1+\e_2}}{\oF_{1,\e_1+\e_2}}=1-b_1-b_2,
\]
which is independent of $\{a_i\}$.
The sum of other terms is dependent on $\{a_i\}$.
Thus $\of_{4,\pm \e_1 \pm \e_2} \ne 0$.
\item In other cases we have $\frac{\of_{2,\n}}{\oF_{1,\n}}=1$.
Other terms explicitly depend on the parameters $\{a_i,\,b_i\}$.
Thus $f_{4,\n} \ne 0$．
\end{itemize}
\end{step}

\begin{step}[irreducibility for $t\ge 5$]
In the case of $t=5$ we have the following two factorizations:
\[
f_{5,\n}=\left(\prod_{\m}f_{0,\m}^{\alpha_{\m}} \right)f_{5,\n}'=
\left(\prod_{r=1}^4\prod_{\m}f_{r,\m}^{\beta_{\m}^{(r)}} \right)f_{5,\n}''
\quad (\alpha_{\m},\,\beta_{\m}^{(r)} \in \Z_{\ge 0},\, f_{5,\n}',f_{5,\n}''\; (\mbox{irreducible})).
\]
Suppose that $f_{5,\n}$ is {\em not} irreducible.
Since $f_{t,\n}$ ($t \le 4$)  is an irreducible non-unit, only the following type of factorization is possible:
\[
f_{5,\n}=u f_{0,\m_0}f_{r,\m_r} \quad (r \in \{1,2,3,4\}),
\]
where $u$ is a unit element.
Now let us recall the discussion in Claims \ref{deg_prop1} and \ref{deg_prop2}.
There exists $r \in \{1,2,3,4\}$ such that $y_5=y_0y_r$.
Therefore $d_5=d_r+d_0$.
On the other hand, we have
$d_{-1}=0, d_0=1$ and
\[
d_{t+1}=Nd_t-d_{t-1}+\frac{1-(-1)^t}{2}.
\]
Thus $d_{t+1} \ge (N-1)d_{t}$ ($t \ge 0$) and therefore
$d_r+d_0 \le d_4+d_0<d_5$,
which is a contradiction.
Thus $f_{t,\n}$ must be irreducible.
The case of $t \ge 6$ is proved in the same manner since $d_t>d_0+d_r$ ($1\le r\le t-1$).

The proof of the irreducibility is now completed and thus the proof of Theorem \ref{th-exdKdV_f} is finished.
\end{step}
%
%

%
%
%

\section{Conclusion}
In this paper, we introduced the equation \eqref{pedKdV_eq} as an example of a coprimeness-preserving non-integrable equation defined over the integer lattice of arbitrary dimension.
The equation \eqref{pedKdV_eq} is a higher dimensional equation of the Hietarinta-Viallet type: i.e., confining but having exponential degree growth.
We investigated the generalized tau-function form \eqref{poly-form} of \eqref{pedKdV_eq}
and showed its Laurent, the irreducibility and the coprimeness properties.
The proof requires heavy calculation but relies only on elementary facts on the factorization of Laurent polynomials.
The equation \eqref{pedKdV_eq} gives several known coprimeness-preserving equations including the coprimeness-preserving non-integrable extension to the disrete KdV equation \eqref{eq:nonlinear} and its generalized tau-function form \eqref{eq:laurent}, whose properties are reviewed in the Appendix.

A reduction of a discrete system gives an equation on a lower dimensional lattice.
It is expected to obtain equations with interesting properties (such as the Laurent property) by reductions from \eqref{pedKdV_eq}.
For example,
a reduction to a one-dimensional lattice of the equation
\eqref{eq:nonlinear} has the coprimeness property and its algebraic entropy can be derived using this property \cite{RIMSentropy}.
Regarding the Laurent property, it is known that if an equation satisfies the Laurent property
for  arbitrary ``good'' domains, it preserves the Laurent property
under reduction \cite{investigation}.
In this paper, however, we proved the Laurent property of \eqref{poly-form} only on a specific domain.
We have some results on problems arising from the domain of definition, which we wish to present in future works.

\section*{Acknowledgment}

We thank Prof.~R.~Willox for valuable comments.
This work is partially supported by KAKENHI grant numbers 17K14211, 18H01127 and 18K13438.

\appendix


\section{Factorization of Laurent polynomials}

We introduce Lemma \ref{lemma-common}, which characterizes the factorization of Laurent polynomials under the transformation of variables.
The statement was first introduced in \cite{coprime} when the ring of coefficients is $R=\mathbb{Z}$,
and later, $\mathbb{Z}$ is extended to an arbitrary unique factorization domain (UFD) $R$ as follows:
\begin{lemma}[\cite{NonlinToda}] \label{lemma-common}
Let $R$ be a UFD.
Suppose that we have a bijective mapping between $\boldsymbol{p}:=\{p_1,p_2,...,p_N\}$ and $\boldsymbol{q}:=\{q_1,q_2,...,q_N\}$ such that $\boldsymbol{q} \subset R[\boldsymbol{p}^\pm]$, $\boldsymbol{p} \subset R[\boldsymbol{q}^\pm]$,
and that each $q_i$  is irreducible in $R[\boldsymbol{p}^{\pm}]$ for $j=1,2,\cdots, N$.
Here $\boldsymbol{p}^\pm$ denotes $\{p_1^{\pm}, p_2^{\pm},\cdots ,p_N^{\pm}\}$ and so on.
Let us take a Laurent polynomial $f(\boldsymbol{q}) \in R[\boldsymbol{q}^\pm]$ which
is irreducible in $R[\boldsymbol{p}^\pm]:$ i.e., $\tilde{f}(\boldsymbol{p}):=f(\boldsymbol{q}(\boldsymbol{p}))\in R[\boldsymbol{p}^\pm]$ is irreducible.
Then $f(\boldsymbol{q})$ admits the following factorization in $R[\boldsymbol{q}^\pm]$:
\[
f(\boldsymbol{q})=\boldsymbol{p}(\boldsymbol{q})^{\boldsymbol{\alpha}} \tilde{f}(\boldsymbol{q}) \qquad (\boldsymbol{\alpha}\in \Z^N,\, \tilde{f}\ \mbox{is irreducible in} \ R[\boldsymbol{q}^\pm]).
\]
\end{lemma}

The following lemma is a generalization of Lemma~\ref{lemma-common}, which is a basic result in the elementary ring theory.
\begin{lemma}[\cite{NonlinToda}] \label{lemma-common-ufd}
Let $A$ be a UFD and let $S \subset A$ be multiplicative, i.e.,\ $1 \in S$, $0 \notin S$, $s t \in S$ for every $s, t \in S$.
Then, the localization $S^{-1} A := \{ \frac{a}{s} \mid a \in A, s \in S \}$ is a UFD and any irreducible element in $A$ is also irreducible in $S^{-1} A$.
\end{lemma}


\section{The Laurent and the irreducibility properties of Equation \eqref{eq:laurent}}\label{app:2dim}

In this section, we review the Laurent and the irreducibility properties of Equation \eqref{eq:laurent} for a general good domain (Theorem~\ref{thm:laurent}),
which first appeared in  \cite{RIMSJPN}.
Note that in this section, the coefficient ring $R$ (resp.\ the parameters $a$, $b$) can be taken as any UFD (resp.\ any nonzero values in $R$),
while the discussion in \textsection \ref{sec:main} relies on the specific choices of $R$ and $a_i,b_i$: $R=\mathbb{Z}[a_i,b_i]_{1\le i\le d}$ and the parameters $a_i,b_i$ are indeterminate variables.
\begin{definition}(good domain \cite{rims,investigation})\label{def:gooddomain}
Let us denote by ``$\le$'' the product order on the lattice $\mathbb{Z}^2$, i.e.
\[
	h \le h'
	\quad
	\Leftrightarrow
	\quad
	t \le t' \text{ and } n \le n' 
	\qquad
\]
for $h = (t, n), h' = (t', n') \in \mathbb{Z}^2$.
A nonempty subset $H \subset \mathbb{Z}^2$ is a good domain (with respect to Equation \eqref{eq:laurent}) if it satisfies the following two conditions:
\begin{itemize}
\item
If $(t, n) \in H$, then $(t+1, n), (t, n+1) \in H$.
\item
For any $h \in H$, the set
\[
	\{ h' \in H \mid h' \le h \}
\]
is finite.
\end{itemize}
For a good domain $H \subset \mathbb{Z}^2$, we define
\[
	H_0 = \{ (t, n) \in H \mid (t-2, n-2) \notin H \},
\]
which we call the initial domain for $H$.
\end{definition}

\begin{remark}[\cite{HMW}]
The first condition on a good domain requires that the intersection between $H$ and the past light-cone
emanating from $h\in H$ is a finite set.
The second condition on a good domain requires that the future light-cone emanating from $h\in H$
does not intersect with $H_0$.
\end{remark}

\begin{example}\label{ex:ldomain}
The first quadrant
\[
	H = \{ (t, n) \in \mathbb{Z}^2 \mid t, n \ge 0 \}
\]
is a good domain.
The corresponding initial domain is given by
\[
	H_0 = \{ (t, n) \in H \mid t = 0, 1,\, n = 0, 1 \}.
\]
This domain plays an important role in the proof of Theorem~\ref{thm:laurent}.
\end{example}

\begin{definition}[$d_H(h)$]
For a good domain $H \subset \mathbb{Z}^2$, we define the function $d_H \colon H \to \mathbb{Z}_{> 0}$ by
\[
	d_H(h) = \# \{ h' \in H \mid h' \le h \},
\]
where ``$\#$'' denotes the cardinality of a set.
\end{definition}

\begin{lemma}
If $h_1, h_2 \in H$ satisfy $h_1 \le h_2$, then we have $d_H(h_1) \le d_H(h_2)$ and
the equality holds only if $h_1 = h_2$.
\end{lemma}
\begin{proof}
Let $h_1 \le h_2$.
Then we have
\[
	\{ h \in H \mid h \le h_1 \} \subset \{ h \in H \mid h \le h_2 \}.
\]
Taking the cardinality we obtain $d_H(h_1) \le d_H(h_2)$.
If $h_1 \ne h_2$, then $h_2$ belongs only to the right hand side and thus we have $d_H(h_1) \ne d_H(h_2)$.
\end{proof}

\begin{definition}
Let $f$ be a Laurent polynomial.
Then, $f$ can be uniquely factorized as
\[
	f = g h,
\]
where $g$ is a monic Laurent monomial and $h$ is a polynomial without any monomial factor.
We shall call $h$ the polynomial part of $f$.
\end{definition}

\begin{theorem}\label{thm:laurent}
Let $R$ be a UFD and let $a, b \in R$ be nonzero.
Then, Equation \eqref{eq:laurent} has the Laurent property on any good domain, i.e.\ every iterate is a Laurent polynomial of the initial variables.
Moreover, every iterate is irreducible as a Laurent polynomial.
\end{theorem}

\begin{proof}
Let $H \subset \mathbb{Z}^2$ be an arbitrary good domain.
Let $A$ be the Laurent polynomial ring of the initial variables, i.e.
\[
	A = R \left[ f^{\pm}_{h_0} \mid h_0 \in H_0 \right],
\]
where $f_{h_0}$ ($h_0 \in H_0$) is a initial variable.
Let us show by induction on $d_H(h)$ that
\begin{itemize}
\item
$f_h \in A$,
\item
$f_h$ is irreducible as an element of $A$.
\end{itemize}
Note that we do not fix $H$ in the proof.
For example, if another good domain $H'$ and its element $h' \in H'$ satisfy $d_{H'}(h') < d_H(h)$, then we assume in the induction step that $f_{h'}$ belongs to $R [ f^{\pm}_{h'_0} \mid h'_0 \in H'_0 ]$ and is irreducible.

\setcounter{step}{0}
\begin{step}
If $h \in H_0$, then $f_h$ is an initial variable.
Therefore, $f_h \in A$ and $f_h$ is irreducible.
From here on, we only consider the case $h \in H \setminus H_0$.
\end{step}

\begin{step}
Let us show that $f_h \in A$.
Let
\[
	F = - f_{t-2,n-2} f^k_{t-1,n} f^k_{t,n-1}
		+ a f^{k^2-1}_{t-1,n-1} f^{k^2}_{t,n-2} f^k_{t-1,n} f^k_{t-2,n-1}
		+ b f^{k^2-1}_{t-1,n-1} f^{k^2}_{t-2,n} f^k_{t,n-1} f^k_{t-1,n-2}.
\]
Since $f_h = \dfrac{F}{f^k_{t-2,n-1}f^k_{t-1,n-2}}$, it is sufficient to show that $F$ is divisible by $f^k_{t-2,n-1} f^k_{t-1,n-2}$ in the ring $A$.
By the induction hypothesis, $f_{t-2,n-1}$ and $f_{t-1,n-2}$ are both irreducible.
Since there is at least one initial variable contained in only one of them, $f_{t-2,n-1}$ and $f_{t-1,n-2}$ are coprime with each other.
Therefore, it is sufficient to show that $F$ is divisible by both $f^k_{t-2,n-1}$ and $f^k_{t-1,n-2}$, respectively.
Because of the symmetry between $t$ and $n$, it is sufficient to show that $F$ can be divided by $f^k_{t-2,n-1}$.

If $(t-2, n-1) \in H_0$, then $f_{t-2,n-1}$ is an initial variable.
Thus, we only consider the case $(t-2, n-1) \notin H_0$.
Since $H$ is a good domain, neither $f_{t-1,n}$ nor $f_{t-1,n-1}$ is an initial variable.
Since
\[
	F \equiv f^k_{t,n-1} \left( - f_{t-2,n-2}f^k_{t-1,n} + b f^{k^2-1}_{t-1,n-1}f^{k^2}_{t-2,n}f^k_{t-1,n-2} \right) \mod f^k_{t-2,n-1},
\]
it is sufficient to show that
\[
	F' := - f_{t-2,n-2}f^k_{t-1,n} + b f^{k^2-1}_{t-1,n-1}f^{k^2}_{t-2,n}f^k_{t-1,n-2} \equiv 0 \mod f^k_{t-2,n-1}.
\]
By the induction hypothesis, $f_{t-2,n-1}$ is coprime with $f_{t-3,t-1}$, $f_{t-2,t-2}$ and $f_{t-2,t-3}$.
Therefore, it is sufficient to show that $F' = 0$ in the ring
\[
	A[f^{-1}_{t-3,n-1}, f^{-1}_{t-2,n-2}, f^{-1}_{t-2,n-3}]/ \left( f^k_{t-2,n-1} \right).
\]
From here on, all calculations in this step will be done in this ring.

Since $f^k_{t-2,n-1} = 0$, we have
\[
	f_{t-1,n} = - \frac{f_{t-3,n-2}f^k_{t-2,n}f^k_{t-1,n-1}}{f^k_{t-3,n-1}f^k_{t-2,n-2}}
\]
and
\begin{align*}
	F' &=
	-\frac{f^k_{t-3,n-2} f^{k^2}_{t-2,n} f^{k^2}_{t-1,n-1}}{f^{k^2}_{t-3,n-1} f^{k^2-1}_{t-2,n-2}}
	+ b f^{k^2-1}_{t-1,n-1} f^{k^2}_{t-2,n} f^k_{t-1,n-2} \\
	&= \frac{f^{k^2}_{t-2,n} f^{k^2-1}_{t-1,n-1}}{f^{k^2}_{t-3,n-1} f^{k^2-1}_{t-2,n-2}} \left( - f^k_{t-3,n-2} f_{t-1,n-1} + b f^k_{t-1,n-2} f^{k^2}_{t-3,n-1} f^{k^2-1}_{t-2,n-2} \right) \\
	&= \frac{f^{k^2}_{t-2,n} f^{k^2-1}_{t-1,n-1}}{f^{k^2}_{t-3,n-1} f^{k^2-1}_{t-2,n-2} f^k_{t-2,n-3}} \left( - f^k_{t-3,n-2} f_{t-1,n-1} f^k_{t-2,n-3} + b f^k_{t-1,n-2} f^{k^2}_{t-3,n-1} f^{k^2-1}_{t-2,n-2} f^k_{t-2,n-3} \right).
\end{align*}
It follows from Equation (\ref{eq:laurent}) that
\begin{align*}
	& - f^k_{t-3,n-2}f_{t-1,n-1}f^k_{t-2,n-3} + bf^k_{t-1,n-2}f^{k^2}_{t-3,n-1}f^{k^2-1}_{t-2,n-2}f^k_{t-2,n-3} \\
	&=
	- f_{t-3,n-3} f^k_{t-2,n-1} f^k_{t-1,n-2} + a f^{k^2-1}_{t-2,n-2} f^{k^2}_{t-1,n-3} f^k_{t-2,n-1} f^k_{t-3,n-2} \\
	&= 0
\end{align*}
and thus we have $F' = 0$.

We have proved the Laurent property of $f_h$.
Let us show the irreducibility of $f_h$ in the steps below.
\end{step}

\begin{step}\label{step:fhft-1n}
We show that $f_h$ cannot be divided by $f_{t-1,n}$, $f_{t,n-1}$ or $f_{t-1,n-1}$ if none of these is an initial variable.

Using Equation (\ref{eq:laurent}) we have
\[
	f_{t,n}f^k_{t-2,n-1}f^k_{t-1,n-2} =
	- f_{t-2,n-2} f^k_{t-1,n} f^k_{t,n-1}
	+ a f^{k^2-1}_{t-1,n-1} f^{k^2}_{t,n-2} f^k_{t-1,n} f^k_{t-2,n-1}
	+ b f^{k^2-1}_{t-1,n-1} f^{k^2}_{t-2,n} f^k_{t,n-1} f^k_{t-1,n-2}.
\]
Suppose that $f_{t-1,n}$ divides $f_h$.
Then, considering the both sides modulo $f_{t-1,n}$ we have
\[
	b f^{k^2-1}_{t-1,n-1} f^{k^2}_{t-2,n} f^k_{t,n-1} f^k_{t-1,n-2} \equiv 0 \mod f_{t-1,n},
\]
which leads to a contradiction since $f_{t-1,n-1}, f_{t-2,n}, f_{t,n-1}, f_{t-1,n-2}$ are all coprime with $f_{t-1,n}$.
We can show that $f_h$ is not divisible by $f_{t,n-1}$ or $f_{t-1,n-1}$ in the same way.
\end{step}

\begin{step}
Let us define the set $S$ by
\[
	S = \{ h' \in H \mid h' \le h \}.
\]
In this step we show that if $S$ has at least two minimal elements (with respect to the order $\le$), then $f_h$ is irreducible in $A$.

Let $\left\{ h^{(i)} = (t^{(i)}, n^{(i)}) \right\}_{i=1, \ldots, N} \subset S$ be the set of minimal elements of $S$ ($N \ge 2$).
We define the domains $H^{(i)} \subset \mathbb{Z}^2$ by
\[
	H^{(i)} = H \setminus \{ h^{(i)} \}
\]
for $i = 1, \ldots, N$.
Since $h^{(i)}$ is a minimal element of $H$, $H^{(i)}$ is a good domain.
Since
\[
	d_{H^{(i)}}(h) = d_{H}(h) - 1,
\]
it follows from the induction hypothesis that
\[
	f_h \in A^{(i)} := R \left[ f^{\pm}_{h_0} \mid h_0 \in H^{(i)}_0 \right]
\]
and $f_h$ is irreducible in $A^{(i)}$, where $H^{(i)}_0$ is the initial domain for $H^{(i)}$ and is written as
\[
	H^{(i)}_0 = \{ (t^{(i)} + 2, n^{(i)} + 2) \} \cup H_0 \setminus \{ h^{(i)} \}.
\]

Let us consider the following relations on the localized rings:
\[
	A \subset
	A \left[ f^{-1}_{t^{(i)}+2,n^{(i)}+2} \right] =
	A^{(i)} \left[ f^{-1}_{h^{(i)}} \right] \supset
	A^{(i)}.
\]
Since $f_h$ is irreducible in $A^{(i)}$ and the localization always preserves the irreducibility by Lemma \ref{lemma-common-ufd}, $f_h$ is also irreducible in $A \left[ f^{-1}_{t^{(i)}+2,n^{(i)}+2} \right]$.
Therefore, we can express $f_h$ in $N$-ways as
\begin{equation}\label{eq:irredrep}
	f_h = f^{r_i}_{t^{(i)}+2,n^{(i)}+2} F^{(i)},
\end{equation}
where $F^{(i)} \in A$ is irreducible and $r_i$ is a nonnegative integer.
Since $N \ge 2$, $f_h$ becomes reducible only in the case where $N = 2$ and
\begin{equation}\label{eq:exc}
	f_h = u f_{t^{(1)}+2,n^{(1)}+2} f_{t^{(2)}+2,n^{(2)}+2}
\end{equation}
for a unit $u \in A$.
Let us exclude this case.

If $t^{(i)} + 2 > t$ or $n^{(i)} + 2 > n$, then $f_h$ does not contain the initial variable $f_{h^{(i)}}$.
For example, if $t^{(i)} + 2 > t$, it follows from the minimality of $h_i \in S$ that the four elements
\[
	(t^{(i)}, n^{(i)} + 1), \quad
	(t^{(i)}, n^{(i)} + 2), \quad
	(t^{(i)} + 1, n^{(i)} + 1), \quad
	(t^{(i)} + 1, n^{(i)} + 2)
\]
all belong to $H_0$.
Therefore, if $f_{t',n'}$ contains the initial variable $f_{h^{(i)}}$ for $t' \le t$, then $(t', n')$ must be $h^{(i)}$.
Since $f_{h^{(i)}}$ is clearly contained in the polynomial part of $f_{t^{(i)}+2,n^{(i)}+2}$, if $f_h$ does not contain the initial variable $f_{h^{(i)}}$, then $r_i$ in the expression \eqref{eq:irredrep} must be $0$.
Hence, it is sufficient to exclude the case \eqref{eq:exc} only when $(t^{(i)} + 2, n^{(i)} + 2) \le h$ for $i = 1, 2$.

Let $m$ be the least integer that satisfies $(t - m, n) \in H_0$.
Since the initial variable $f_{t-m,n}$ is contained in the polynomial part of $f_h$, at least one of $f_{t^{(1)}+2,n^{(1)}+2}, f_{t^{(2)}+2,n^{(2)}+2}$ must contain $f_{t-m,n}$ in its polynomial part.
We may assume that $f_{t^{(1)}+2,n^{(1)}+2}$ contains $f_{t-m,n}$.
In this case, we have $(t-m, n) \le (t^{(1)}+2, n^{(1)}+2)$.
On the other hand, we have already assumed that $n^{(1)} + 2 \le n$.
Thus, we have $n = n^{(1)} + 2$.

Taking the least $\ell$ that satisfies $(t, n - \ell) \in H_0$ and discussing in the same way as above, we conclude that there exists $i$ such that $t = t^{(i)} + 2$.
If $i = 1$, then we have $h = (t^{(1)}+2, n^{(1)}+2)$.
In this case, $f_{t-2,n-2}, f_{t-2,n-1}, f_{t-2,n}, f_{t-1,n-2}, f_{t-1,n-1}, f_{t-1,n}, f_{t,n-2}, f_{t,n-1}$ are all initial variables.
Thus, $f_h$ itself is irreducible and the decomposition as in \eqref{eq:exc} is impossible.
Hence, we have
\[
	n = n^{(1)} + 2, \quad
	t = t^{(2)} + 2.
\]

If $(t-3, n-3) \in H$, then a minimal element of the set
\[
	\{ h' \in H \mid h' \le (t-3, n-3) \}
\]
is also minimal in $S$, which leads to contradiction since $S$ has only two minimal elements.
Therefore, we have $(t-3, n-3) \notin H$ and $(t-1, n-1) \in H_0$.

It follows from Equation \eqref{eq:laurent} that $f_h$ contains the initial variable $f_{t-1,n-1}$ in its polynomial part.
Because of the decomposition \eqref{eq:exc}, at least one of $f_{t^{(1)}+2,n^{(1)}+2}$ and $f_{t^{(2)}+2,n^{(2)}+2}$ must contain the initial variable $f_{t-1,n-1}$ in its polynomial part.
We may assume that $f_{t^{(1)}+2,n^{(1)}+2}$ contains $f_{t-1,n-1}$.
Since
\begin{align*}
	(t-1, n-1) &\le (t^{(1)}+2,n^{(1)}+2) \le (t, n), \\
	n &= n^{(1)} + 2, \\
	(t^{(1)}+2,n^{(1)}+2) &\ne (t, n),
\end{align*}
the only possible case is
\[
	t^{(1)} = t - 3, \quad
	n^{(1)} = n - 2.
\]
In this case, the decomposition \eqref{eq:exc} is
\[
	f_h = u f_{t-1,n} f_{t,n^{(2)}+2},
\]
which contradicts to Step~\ref{step:fhft-1n}.

In the steps below, we will show that if $S$ has only one minimal element, then $f_h$ is irreducible.
\end{step}

\begin{step}
As the lattice $\mathbb{Z}^2$ can be translated at will, we may assume that the minimal element of $S$ is the origin $(0, 0)$.
In this case, the set $S$ can be expressed as
\[
	S = \{ (t', n') \in \mathbb{Z}^2 \mid 0 \le t' \le t, 0 \le n' \le n \}.
\]
That is, $S$ coincides with the lower left part of the first quadrant (Example~\ref{ex:ldomain}).
Since $f_h$ is determined by $S$, it is sufficient to show the irreducibility of $f_h$ when $H$ is the first quadrant.
\end{step}

\begin{step}
Let
\[
	H' = H \setminus \{ (0, 0) \}.
\]
It follows from the induction hypothesis that $f_h$ is irreducible as an element of the ring $A' := R \left[ f^{\pm}_{h_0} \mid h_0 \in H' \right]$.
Using the relations on the localized rings
\[
	A \subset A \left[ f^{-1}_{22} \right] = A' \left[ f^{-1}_{00} \right] \supset A',
\]
we can express $f_h$ as
\begin{equation}\label{eq:f22}
	f_h = f^{r}_{22} F
\end{equation}
where $F \in A$ is irreducible and $r \ge 0$.
Therefore, it is sufficient to show that $f_{22}$ does not divide $f_h$ in the ring $A$.
We already showed in Step~\ref{step:fhft-1n} that $f_{22}$ does not divide $f_{23}, f_{32}, f_{33}$.
\end{step}

\begin{step}
Whether $f_{22}$ divides $f_h$ or not is invariant under extension of the coefficient ring.
Therefore, replacing $R$ with the algebraic closure of its field of fractions if necessary, we may assume that $R$ is an algebraically closed field.
\end{step}

\begin{step}\label{step:f22f24}
Let us show that $f_{22}$ does not divide $f_{24}$.
We will substitute appropriate nonzero values (elements of $R$) for some initial variables to check that $f_{22}$ does not divide $f_{24}$.

Let us take the following initial values (see Table~\ref{table:value8}):
\begin{align*}
	f_{01} &= f_{02} = f_{10} = f_{11} = f_{12} = f_{13} = f_{14} = f_{21} = 1, \\
	f_{03} &= b^{-1/k^2}, \quad
	f_{04} = \gamma, \quad
	f_{20} = \delta,
\end{align*}
where $\gamma, \delta \in R$ satisfy
\[
	\gamma^{k^2} \ne \frac{1}{b}, \quad
	\delta^{k^2} \ne -\frac{b}{a}.
\]
Note that we can always take such $\gamma, \delta$ since $R$ is an algebraically closed field.
\begin{table}
\[
\begin{array}{|c|c||c|}\hline
\gamma & 1 & \\ \hline
b^{-1/k^2} & 1 & \\ \hline
1 & 1 & \\ \hline \hline
1 & 1 & 1 \\ \hline
f_{00} & 1 & \delta \\ \hline
\end{array}
\]
\caption{Table showing the initial values in Step~\ref{step:f22f24}.}
\label{table:value8}
\end{table}
Under these initial values, we have
\[
	f_{22} = - f_{00} + a \delta^{k^2} + b.
\]
Since $a \delta^{k^2} + b \ne 0$, $f_{22}$ is of degree one with respect to $f_{00}$ and thus is not a unit in $A$.
Therefore, it is sufficient to show that $f_{24}$ is not divisible by $\epsilon := f_{22}$.

By construction we have $f_{23} = a$.
A direct calculation shows that
\[
	f_{24} = \frac{1}{b^{-\frac{1}{k}}} \left(
	- a^k
	+ a \epsilon^{k^2} b^{-\frac{1}{k}}
	+ b \gamma^{k^2} a^k
	\right)
	= a^k b^{\frac{1}{k}} \left( b \gamma^{k^2} - 1 \right) + a \epsilon^{k^2}
\]
and thus $f_{22}$ does not divide $f_{24}$.

We can prove in the same way that $f_{22}$ does not divide $f_{42}$.
\end{step}

\begin{step}\label{step:f22f34}
We show that $f_{22}$ does not divide $f_{34}$.

Let us take the following initial values (see Table~\ref{table:value9}):
\begin{align*}
	f_{01} &= f_{02} = f_{10} = f_{11} = f_{12} = f_{13} = f_{14} = f_{21} = f_{31} = 1, \\
	f_{03} &= f_{04} = b^{-1/k^2}, \quad
	f_{20} = \delta, \quad
	f_{30} = a^{-1/k^2},
\end{align*}
where $\delta \in R$ satisfies
\[
	\delta^{k^2} \ne -\frac{b}{a}.
\]
\begin{table}
\[
\begin{array}{|c|c||c|c|c|}\hline
b^{-\frac{1}{k^2}} & 1 & \mathcal{O}(\epsilon^{k^2}) & a^{2k^2-1}b^{k^2+1} + \mathcal{O}(\epsilon^{k^2-1}) & \\ \hline
b^{-\frac{1}{k^2}} & 1 & a & - a^k b^k + \mathcal{O}(\epsilon^{k^2-1}) & f_{43} \\ \hline
1 & 1 & \epsilon & b & f_{42} \\ \hline \hline
1 & 1 & 1 & 1 & 1 \\ \hline
f_{00} & 1 & \delta & a^{-\frac{1}{k^2}} & \gamma \\ \hline
\end{array}
\]
\caption{Table showing the calculations in Steps~\ref{step:f22f34} and \ref{step:f22f44}.}
\label{table:value9}
\end{table}
As in Step~\ref{step:f22f24}, $f_{22}$ is of degree one with respect to $f_{00}$ and thus is not a unit in $A$.
Therefore, it is sufficient to show that $f_{34}$ is not divisible by $\epsilon := f_{22}$.

Using $f_{23} = a$, we have
\[
	f_{24} = a \epsilon^{k^2} = \mathcal{O}(\epsilon^{k^2}).
\]
A straightforward calculation shows that
\begin{align*}
	f_{33} &= - a^k b^k + a \epsilon^{k^2-1} a^k + b \epsilon^{k^2-1} b^k \\
	&= - a^k b^k + \mathcal{O}(\epsilon^{k^2-1}), \\
	f_{34} 
	&= \frac{1}{\epsilon^k} \Bigg(
	- \left( \mathcal{O}(\epsilon^{k^2}) \right)^k \left( - a^k b^k + \mathcal{O}(\epsilon^{k^2-1}) \right)^k
	+ a a^{k^2-1} b^{k^2} \left( \mathcal{O}(\epsilon^{k^2}) \right)^k \\
	& \quad\quad+ b a^{k^2-1} \left( - a^k b^k + \mathcal{O}(\epsilon^{k^2-1}) \right)^k \epsilon^k
	\Bigg) \\
	&= a^{2k^2-1}b^{k^2+1} + \mathcal{O}(\epsilon^{k^2-1}).
\end{align*}
Hence, $f_{34}$ is not divisible by $f_{22}$.

We can prove in the same way that $f_{22}$ does not divide $f_{43}$.
\end{step}

\begin{step}\label{step:f22f44}
We show that $f_{22}$ does not divide $f_{44}$.

First we take the same initial values as in Step~\ref{step:f22f34}.
Then, we take the initial values $f_{40}, f_{41}$ as
\[
	f_{40} = \gamma, \quad
	f_{41} = 1,
\]
where $\gamma \in R$ satisfies
\[
	a \gamma^{k^2} \ne \delta.
\]
By construction, we have
\[
	f_{42} = \frac{1}{a^{-\frac{1}{k}}}
	\left(
	- \delta b^k
	+ a \gamma^{k^2} b^k
	+ b \epsilon^{k^2} a^{-\frac{1}{k}}
	\right),
\]
which is not divisible by $\epsilon := f_{22}$.
Using the calculations in Step~\ref{step:f22f34} (see Table~\ref{table:value9}), we have
\begin{align*}
	f_{44} &=
	\frac{1}{a^k b^k}\Bigg(
	- \epsilon \left( a^{2k^2-1}b^{k^2+1} + \mathcal{O}(\epsilon^{k^2-1}) \right)^k f_{43}^k & \\
	& \quad\quad + a \left( - a^k b^k + \mathcal{O}(\epsilon^{k^2-1}) \right)^{k^2-1} f^{k^2}_{42} \left( a^{2k^2-1}b^{k^2+1} + \mathcal{O}(\epsilon^{k^2-1}) \right)^k a^k & \\
	& \quad\quad + b \left( - a^k b^k + \mathcal{O}(\epsilon^{k^2-1}) \right)^{k^2-1} \left( \mathcal{O}(\epsilon^{k^2}) \right)^{k^2} f^k_{43} b^k
	\Bigg) & \\
	&= - a^{3k^3-2k+1} b^{2k^3-k} f^{k^2}_{42} + \mathcal{O}(\epsilon).
\end{align*}
Therefore, $\epsilon$ does not divide $f_{44}$ since $f_{42}$ is not divisible by $\epsilon$.
\end{step}

\begin{step}
Let us show the irreducibility for the remaining cases, i.e.\ the case where $t$ or $n$ is greater than $4$.
In this case, $(t - i, n - j) \ne (2, 2)$ for $i, j = 0, 1, 2$.

Let
\[
	H' = H \setminus \{ (i, j) \mid i, j = 0, 1, 2 \}.
\]
Then, $H'$ is clearly a good domain.
We define two sets $E, E'$ by
\begin{align*}
	E &= \{ (0, 0), (0, 1), (0, 2), (1, 0), (1, 1), (1, 2), (2, 0), (2, 1) \}, \\
	E' &= \{ (2, 3), (2, 4), (3, 2), (3, 3), (3, 4), (4, 2), (4, 3), (4, 4) \}.
\end{align*}
Then we have
\[
	H_0 \setminus H'_0 = E, \quad
	H'_0 \setminus H_0 = E'.
\]
It is important to note that $(2, 2) \notin E, E'$ and $E \cap E' = \emptyset$.

Since $d_{H'}(h) < d_H(h)$, it follows from the induction hypothesis that $f_h$ is irreducible as an element of the ring
\[
	A' := R \left[ f^{\pm}_{h_0} \mid h_0 \in H' \right].
\]
Using the relation on the localized rings
\[
	A \subset A \left[ f^{-1}_{h_0} \mid h_0 \in E' \right] = A' \left[ f^{-1}_{h_0} \mid h_0 \in E \right] \supset A',
\]
we can express $f_h$ as
\[
	f_h = F' \prod_{(i, j) \in E'} f^{r'_{ij}}_{ij},
\]
where $F' \in A$ is irreducible and $r'_{ij} \ge 0$ ($(i,j) \in E'$).
Since $(2, 2) \notin E'$, it follows from the decomposition \eqref{eq:f22} that, if $f_h$ is not irreducible, then it must be decomposed as
\begin{equation}\label{eq:imp}
	f_h = u f_{22} f_{ij},
\end{equation}
where $u \in A$ is a unit and $(i, j) \in E'$.

In this step we have assumed that at least one of $t, n$ is greater than $4$.
Assume $t \ge 5$ and take the least integer $m$ that satisfies $(t, n - m) \in H_0$.
Then, the initial variable $f_{t,n-m}$ is contained in the polynomial part of $f_h$.
However, neither $f_{22}$ nor $f_{ij}$ contains $f_{t,n-m}$ since $i \le 4$.
Therefore, the decomposition \eqref{eq:imp} is impossible.
\end{step}
\end{proof}

\begin{corollary}
Equation \eqref{eq:laurent} has the coprimeness property on any good domain, i.e.\ every pair of the iterates is coprime as Laurent polynomials of the initial variables.
\end{corollary}

%
%
%
%

\section{Coprimeness property of Equation \eqref{eq:nonlinear}}\label{sec:xcoprime}

In this section, we prove the coprimeness property of Equation \eqref{eq:nonlinear}.

\begin{definition}\label{def:nonlineargooddomain}
A nonempty subset $G \subset \mathbb{Z}^2$ is a good domain (with respect to Equation \eqref{eq:nonlinear}) if it satisfies the following two conditions:
\begin{itemize}
\item
If $(t, n) \in G$, then $(t+1, n), (t, n+1) \in G$.
\item
For any $h \in G$, the set
\[
	\{ h' \in G \mid h' \le h \}
\]
is finite.
\end{itemize}
For a good domain $G \subset \mathbb{Z}^2$, we define
\[
	G_0 = \{ (t, n) \in G \mid (t-1, n-1) \notin G \},
\]
which we call the initial domain for $G$.
\end{definition}

Note that the only difference between Definitions~\ref{def:gooddomain} and \ref{def:nonlineargooddomain} is the definition of the initial domain.

\begin{theorem}\label{thm:nonlinearcoprimeness}
Let $G \subset \mathbb{Z}^2$ be a good domain with respect to Equation \eqref{eq:nonlinear} and we consider each iterate $x_h$ as a rational function of the initial variables $x_{h_0}$ ($h_0 \in G_0$).
Then, there exists a family of irreducible Laurent polynomials $f'_{h}$ ($h \in H$) of the initial variables such that each $x_h$ is decomposed as
\[
	x_{t,n} = \frac{f'_{t,n}f'_{t-1,n-1}}{f'^k_{t-1,n}f'^k_{t,n-1}}.
\]
Moreover, $f'_h$ and $f'_{h'}$ are coprime as Laurent polynomials unless $h = h'$.
In particular, we have
\[
	|t - t'| > 1 \text{ or } |n - n'| > 1 \quad
	\Rightarrow
	\quad
	x_{t,n} \text{ and } x_{t',n'} \text{ are coprime as rational functions},
\]
i.e.\ Equation \eqref{eq:nonlinear} satisfies the coprimeness property on any good domain.
\end{theorem}

\begin{proof}
\setcounter{step}{0}
\begin{step}
Let $H \subset \mathbb{Z}^2$ be the set obtained by translating $G$ in the $(-1, -1)$-direction, i.e.
\[
	H = G - (1, 1).
\]
It is clear by Definitions~\ref{def:gooddomain} and \ref{def:nonlineargooddomain} that $H$ is a good domain with respect to Equation \eqref{eq:laurent} and its initial domain is
\[
	H_0 = G_0 \cup \left( G_0 - (1, 1) \right).
\]
Therefore, if $(t_0, n_0) \in G_0$, then the 4 points
\[
	(t_0, n_0), \quad
	(t_0 - 1, n_0), \quad
	(t_0, n_0 - 1), \quad
	(t_0 - 1, n_0 - 1)
\]
all belong to $H_0$.
Let us consider Equation \eqref{eq:laurent} on $H$.
\end{step}

\begin{step}
Let us take appropriate initial values for $f_{h_0}$ ($h_0 \in H_0$) so that
\[
	x_{t_0,n_0} = \frac{f_{t_0,n_0}f_{t_0-1,n_0-1}}{f^k_{t_0-1,n_0}f^k_{t_0,n_0-1}}
\]
for each $(t_0, n_0) \in G_0$.

Let us introduce auxiliary variables $Y_{h_0}$ for each $h_0 \in H_0 \setminus G_0$.
We consider the system of equations with respect to the (infinitely many) variables $f'_{t,n}$ ($(t,n) \in H_0$)
\begin{equation}\label{eq:ftox}
	f'_{t,n} = \begin{cases}
		Y_{h} & (h \in H_0 \setminus G_0) \\
		\dfrac{x_{t,n}f'^k_{t-1,n}f'^k_{t,n-1}}{f'_{t-1,n-1}} & (h \in G_0).
	\end{cases}
\end{equation}
Since the indices of $f'$ on the right hand side are all smaller than $(t, n)$ with respect to the product order $\le$, we can solve this system and we obtain that each $f_{h_0}$ is a monic Laurent monomial of $x_{h_0}$ ($h_0 \in G_0$) and $Y_{h_0}$ ($h_0 \in H_0 \setminus G_0$).
Let us define $f'_h$ for $h \in H \setminus H_0$ by Equation \eqref{eq:laurent}, i.e.\
\[
	f'_h := \left. f_h \right|_{f_{h_0} \leftarrow f'_{h_0} (h_0 \in H_0)}
	\quad
	(h \in H \setminus H_0),
\]
where each $f_h$ is a Laurent polynomial of $f_{h_0}$ ($h_0 \in H_0$) and $\left\{ f'_{h_0} \right\}_{h_0 \in H_0}$ is the solution of the system \eqref{eq:ftox}.
Then, it is clear by construction that $x_{t,n}$ satisfies Equation \eqref{eq:nonlinear} and
\[
	x_{t,n} = \frac{f'_{t,n} f'_{t-1,n-1}}{f'^k_{t-1,n} f'^k_{t,n-1}}
\]
for all $(t, n) \in G$.
Note that $x_h$ is independent of $Y_{h_0}$ while each $f'$ on the right hand is not.
\end{step}

\begin{step}
Let us show that $f'_h$ is irreducible as a Laurent polynomial of $x_{h_0}, Y_{h_0}$.

The system \eqref{eq:ftox} can be thought of as defining a variable transformation from $\{x_{h_0}, Y_{h_0}\}$ to $\{f'_{h_0}\}$.
It follows from the previous step that this transformation is defined by Laurent monomials.
Its inverse transformation is also defined by monic Laurent monomials as
\[
	\begin{cases}
		Y_{t,n} = f'_{t,n} & ((t,n) \in H_0 \setminus G_0) \\
		x_{t,n} = \dfrac{f'_{t,n} f'_{t-1,n-1}}{f'^k_{t-1,n} f'^k_{t,n-1}} & ((t,n) \in G_0).
	\end{cases}
\]
Therefore, the variable transformation from $\{x_{h_0}, Y_{h_0}\}$ to $\{f'_{h_0}\}$ is given by a ring isomorphism and thus it preserves the irreducibility of each element.
Since every iterate $f'_h$ is irreducible as a Laurent polynomial of $f'_{h_0}$ by Theorem~\ref{thm:laurent}, it is also irreducible as a Laurent polynomial of $x_{h_0}, Y_{h_0}$.
\end{step}
\end{proof}

\end{document}